\documentclass[a4paper,envcountsame]{llncs}

\usepackage{ifthen}

\usepackage{color}
\usepackage{amsmath,amsfonts}
\usepackage{txfonts}

\usepackage{tikz}
\tikzstyle{max}=[shape=rectangle,draw,inner sep=0pt,minimum size=6mm,thick]
\tikzstyle{ran}=[shape=circle,draw,inner sep=0pt,minimum size=6mm,thick]

\newcommand{\Zset}{\mathbb{Z}}
\newcommand{\Rset}{\mathbb{R}}
\newcommand{\A}{\mathcal{A}}
\newcommand{\C}{\mathcal{C}}
\newcommand{\E}{\mathbb{E}}
\newcommand{\G}{\mathcal{G}}
\newcommand{\M}{\mathcal{M}}
\newcommand{\calO}{\mathcal{O}}
\newcommand{\Nset}{\mathbb{N}}
\newcommand{\Qset}{\mathbb{Q}}
\newcommand{\eps}{\varepsilon}
\newcommand{\size}[1]{||#1||}

\newcommand{\B}{\mathcal{B}}

\newcommand{\genTran}[2]{%
    {}\mathchoice%
    {\stackrel{#1}{#2}}
    {\mathop {\smash{#2}}\limits^{\vrule width 0pt height 0pt depth 4pt\smash{#1}}}
    {\stackrel{#1}{#2}}
    {\stackrel{#1}{#2}}
{}}

\newcommand{\btran}[1]{\genTran{#1}{\leadsto}}

\newcommand{\ltran}[1]{\genTran{#1}{\longrightarrow}}
\newcommand{\stochS}[1]{#1_0}  %
\newcommand{\pOneS}[1]{#1_1}  %

\newcommand{\conf}[2]{#1(#2)}  %
\newcommand{\cf}[1]{\mathcal{C}}  %
\newcommand{\len}[1]{len(#1)}
\newcommand{\runs}[1]{Run(#1)}  %

\newcommand{\Prb}[2]{\mathbb{P}^{#1}_{#2}}
\renewcommand{\Pr}[3]{\Prb{#1}{#2}\hspace{-0.16em}\left({#3}\right)}   %

\newcommand{\counter}[1]{C^{(#1)}}
\newcommand{\cstate}[1]{S^{(#1)}}

\newcommand{\Prob}{\mathit{Prob}}

\newcommand{\val}{\mathrm{Val}}

\newcommand{\fast}{}
\newcommand{\height}[1]{\mathit{height}(#1)}
\newcommand{\MEC}{\mathit{MEC}}

\newcommand{\primsol}{\left(\bar{x},(\bar{\poten}_q)_{q\in Q}\right)}
\newcommand{\dualsol}{\left((\bar{\lpfreq}_q)_{q\in Q_0},(\bar{\lpfreq}_{(q,i,q')})_{q \in Q_1,  (q,i,q') \in \delta}\right)}
\newcommand{\lpfreq}{y}

\newcommand{\lowerrem}[3]{B\left(#1,#2,#3\right)}

\newcommand{\trancard}[1]{n_{#1}}
\newcommand{\tleave}{T_{\rightarrow}}
\newcommand{\maralt}[1]{{\hat{m}}^{(#1)}}
\newcommand{\switch}{W}
\newcommand{\reach}{\mathit{hit}}
\newcommand{\leave}{\mathit{lv}}
\newcommand{\fin}{\mathit{fin}}

\newcommand{\sL}{\mathcal{L}}
\newcommand{\mar}[1]{m^{(#1)}}

\newcommand{\poten}{z} %

\newcommand{\ifApp}[2]%
{\ifthenelse{\isundefined{\showappendix}}{#2}{#1}}

\spnewtheorem{fact}[theorem]{Fact}{\bfseries}{\itshape}

\title{Minimizing Expected Termination Time in One-Counter Markov 
Decision Processes}
\titlerunning{Expected Termination Time in OC-MDPs}

\author{
Tom\'{a}\v{s} Br\'{a}zdil\inst{1}\thanks{Tom\'{a}\v{s} Br\'{a}zdil and Petr
Novotn\'{y} are supported by the Czech Science Foundation, 
grant No.~P202/12/G061. Anton\'{\i}n Ku\v{c}era is supported 
by the Czech Science Foundation, grant No.~P202/10/1469. Dominik 
Wojtczak is supported by EPSRC grant EP/G050112/2.} \and
Anton\'{\i}n Ku\v{c}era\inst{1}$^\star$ \and
Petr Novotn\'{y}\inst{1}$^\star$ \and
Dominik Wojtczak\inst{2}$^\star$
}

\authorrunning{
Br\'{a}zdil \and
Ku\v{c}era \and
Novotn\'{y} \and
Wojtczak
}

\institute{
Faculty of Informatics, Masaryk University \\
\email{\{xbrazdil,kucera\}@fi.muni.cz, petr.novotny.mail@gmail.com}
\and
Department of Computer Science, University of Liverpool \\
\email{d.wojtczak@liv.ac.uk}
}
\begin{document}

\maketitle

\begin{abstract}
  We consider the problem of computing the value and an optimal strategy
  for minimizing the expected termination time in one-counter Markov
  decision processes. Since the value may be irrational and an optimal
  strategy may be rather complicated, we concentrate on the problems
  of approximating the value up to a given error $\varepsilon >0$ and
  computing a finite representation of an $\varepsilon$-optimal strategy.
  We show that these
  problems are solvable in exponential time for a given configuration,
  and we also show that they are computationally
  hard in the sense that a polynomial-time approximation algorithm
  cannot exist unless P=NP.
\end{abstract}

\section{Introduction}
\label{sec-intro}

In recent years, a lot of research work has been devoted to the study
of stochastic extensions of various automata-theoretic models such as
pushdown automata, Petri nets, lossy channel systems, and many
others. In this paper we study the class of \emph{one-counter Markov
  decision processes (OC-MDPs)}, which are infinite-state MDPs 
\cite{Puterman:book,FV:book}
generated by finite-state automata operating over a single unbounded 
counter. Intuitively, an OC-MDP is specified by a
finite directed graph $\A$ where the nodes are control states and the
edges correspond to transitions between control states. Each control
state is either stochastic or non-deterministic, which means that the
next edge is chosen either randomly (according to a fixed
probability distribution over the outgoing edges) or by a
controller. Further, each edge either increments, decrements, or
leaves unchanged the current counter value. A \emph{configuration}
$\conf{q}{i}$ of an OC-MDP $\A$ is given by the current control state
$q$ and the current counter value~$i$ (for technical convenience,
we also allow negative counter values, although we are only interested 
in runs where the counter stays non-negative).   
The outgoing transitions of $\conf{q}{i}$ are determined by the edges 
of $\A$ in the natural way.

Previous works on OC-MDPs 
\cite{BBEKW:OC-MDP,BBE:OC-games,BBEK:OC-games-termination-approx} 
considered mainly the objective of \emph{maximizing/minimizing termination
probability}.  We say that a run initiated in a configuration
$\conf{q}{i}$ \emph{terminates} if it visits a configuration with 
zero counter. The goal of the controller
is to play so that the probability
of all terminating runs is maximized (or minimized). In this paper,
we study a related objective of \emph{minimizing the expected termination
time}. Formally, we define a random variable $T$ over the runs of $\A$
such that $T(\omega)$ is equal either to $\infty$ (if the run $\omega$ is 
non-terminating) or
to the number of transitions need to reach a configuration with 
zero counter (if $\omega$ is terminating). The goal of the controller
is to minimize the expectation $\E(T)$. The \emph{value} of 
$\conf{q}{i}$ is the infimum of $\E(T)$ over all 
strategies. It is easy to see that the
controller has a memoryless deterministic strategy which is optimal
(i.e., achieves the value) in every configuration.
However, since OC-MDPs have infinitely many configurations, this does 
not imply that an 
optimal strategy is finitely representable and computable. Further, 
the value itself can be irrational. Therefore, we concentrate on 
the problem of \emph{approximating} the value of a given configuration
up to a given (absolute or relative) error $\varepsilon > 0$, and computing
a strategy which is \emph{$\varepsilon$-optimal} (in both absolute 
and relative sense). Our main results can be summarized as follows:

\begin{itemize}
\item \emph{\bfseries 
  The value and optimal strategy can be effectively approximated up to 
  a given relative/absolute error
  in exponential time.} More precisely, we show that given a 
  OC-MDP $\A$, a configuration $\conf{q}{i}$ of $\A$ where $i \geq 0$,
  and $\varepsilon >0$,
  the value of $\conf{q}{i}$ up to the (relative or absolute) error
  $\varepsilon$ is computable in time exponential in the encoding size
  of $\A$, $i$, and $\varepsilon$, where all numerical constants are 
  represented as fractions of binary numbers. Further, there is 
  a history-dependent deterministic strategy $\sigma$ computable
  in exponential time such that the absolute/relative difference 
  between the value of $\conf{q}{i}$ and the outcome of $\sigma$ in 
  $\conf{q}{i}$ is bounded by $\varepsilon$.
\item \emph{\bfseries The value is not approximable in polynomial time 
  unless P=NP.}
  This hardness result holds even if we restrict ourselves to 
  configurations with counter value equal to~$1$ and to
  OC-MDPs where every outgoing edge of a stochastic 
  control state has probability $1/2$. The result is valid for absolute
  as well as relative approximation.
\end{itemize}

\noindent
Let us sketch the basic ideas behind these results. The upper bounds 
are obtained in two steps. In the first step (Section~\ref{sec:scMDP}), 
we analyze the special case when the underlying graph of~$\A$ is strongly 
connected.
We show that minimizing the expected termination time is closely related
to minimizing the expected increase of the counter per transition, at 
least for large counter values. We start by 
computing the minimal expected increase of the counter per transition
(denoted by $\bar{x}$) achievable by the controller, and the
associated strategy~$\sigma$. This is done by standard linear 
programming techniques developed for optimizing the long-run average 
reward in finite-state MDPs (see, e.g.,~\cite{Puterman:book}) applied 
to the underlying finite graph of~$\A$. Note that $\sigma$ depends only
on the current control state and ignores the current counter value
(we say that $\sigma$ is \emph{counterless}). Further, the encoding size
of $\bar{x}$ is \emph{polynomial} in $\size\A$. Then, we distinguish 
two cases.
\smallskip

\noindent
\textit{Case~(A), $\bar{x} \geq 0$}. Then the counter does not have a
tendency to decrease \emph{regardless} of the controller's strategy,
and the expected termination time value is infinite in all
configurations $\conf{q}{i}$ such that $i \geq |Q|$, where $Q$ is the
set of control states of~$\A$ (see
Proposition~\ref{prop:charact}.~A). For the finitely many remaining
configurations, we can compute the value and optimal strategy
precisely by standard methods for finite-state MDPs.  
\smallskip

\noindent
\textit{Case~(B), $\bar{x} < 0$}. Then, one intuitively expects that
applying the strategy $\sigma$ in an initial configuration
$\conf{q}{i}$ yields the expected termination time about
$i/|\bar{x}|$. Actually, this is \emph{almost} correct; we show
(Proposition~\ref{prop:charact}.~B.2) that this expectation is bounded
by $(i+U)/|\bar{x}|$, where $U \geq 0$ is a constant depending only on
$\A$ whose size is at most exponential in~$\size\A$. Further, we show
that an \emph{arbitrary} strategy $\pi$ applied to $\conf{q}{i}$
yields the expected termination time \emph{at least}
$(i-V)/|\bar{x}|$, where $V \geq 0$ is a constant depending only on
$\A$ whose size is at most exponential in~$\size\A$
(Proposition~\ref{prop:charact}.~B.1). In particular, this applies to
the \emph{optimal} strategy $\pi^*$ for minimizing the expected
termination time. Hence, $\pi^*$ can be more efficient than $\sigma$,
but the difference between their outcomes is bounded by a constant
which depends only on $\A$ and is at most exponential in $\size\A$.
We proceed by computing a sufficiently large $k$ so that the
probability of increasing the counter to $i+k$ by a run initiated in
$\conf{q}{i}$ is inevitably (i.e., under any optimal strategy) 
so small that the controller can safely switch to the
strategy~$\sigma$ when the counter reaches the value $i+k$. Then, we
construct a \emph{finite-state} MDP $\M$ and a reward function $f$
over its transitions such that
\begin{itemize}
\item the states are all configurations $\conf{p}{j}$ where $0 \leq j
  \leq i+k$;
\item all states with counter values less than $i+k$ ``inherit'' their
  transitions from $\A$; configurations of the form $\conf{p}{i+k}$
  have only self-loops;
\item the self-loops on configurations where the counter equals
  $0$ or $i+k$ have zero reward, transitions leading to configurations
  where the counter equals $i+k$ have reward  $(i+k+U)/|\bar{x}|$,
  and the other transitions have reward~$1$.
\end{itemize}
In this finite-state MDP $\M$, we compute an optimal memoryless
deterministic strategy $\varrho$ for the total accumulated reward
objective specified by~$f$. Then, we consider another strategy 
$\hat{\sigma}$ for $\conf{q}{i}$ which behaves like
$\varrho$ until the point when the counter reaches $i+k$, and from
that point on it behaves like~$\sigma$. It turns out that the absolute as
well as relative difference between the outcome of $\hat{\sigma}$ in
$\conf{q}{i}$ and the value of $\conf{q}{i}$ is bounded
by~$\varepsilon$, and hence $\hat{\sigma}$ is the desired
$\varepsilon$-optimal strategy.
\smallskip

In the general case when $\A$ is not necessarily strongly connected
(see Section~\ref{sec:nscMDP}),
we have to solve additional difficulties. Intuitively, we split
the graph of~$\A$ into maximal end components (MECs), where each MEC
can be seen seen as a strongly connected \mbox{OC-MDP} and analyzed by 
the techniques discussed above. In particular, for every MEC~$C$ we
compute the associated $\bar{x}_C$ (see above).
Then, we consider a strategy which tries 
to reach a MEC as quickly as possible so that the expected value
of the fraction $1/|\bar{x}_C|$ is \emph{minimal}. After reaching
a target MEC, the strategy starts to behave as the strategy $\sigma$ discussed
above. It turns out that this particular strategy cannot be much worse
than the optimal strategy (a proof of this claim requires new 
observations), and the rest of the argument is similar as in the
strongly connected case.

The lower bound, i.e., the result saying that the value cannot be
efficiently approximated unless P=NP (see Section~\ref{sec-lower}), 
seems to be the first result of
this kind for OC-MDPs. Here we combine the technique of encoding
propositional assignments presented in
\cite{Kucera:OC-FS-weak-bisimilarity-TCS} (see also
\cite{JKMS:one-counter-generic-IC}) with some new gadgets constructed
specifically for this proof (let us note that we did not manage to
improve the presented lower bound to PSPACE by adapting other known
techniques
\cite{GL:one-counter-branching,Serre:Parity-games-OC,JS:AFA-one-letter-IPL}).
As a byproduct, our proof also reveals that the optimal strategy for
minimizing the expected termination time \emph{cannot} ignore the
precise counter value, even if the counter becomes very large. In our
example, the (only) optimal strategy is \emph{eventually periodic} in
the sense that for a sufficiently large counter value~$i$, it is only
``$i$ modulo $c$'' which matters, where $c$ is a fixed (exponentially
large) constant.  The question whether there \emph{always} exists an
optimal eventually periodic strategy is left open. Another open
question is whether our results can be extended to stochastic games
over one-counter automata.  \smallskip

\noindent
\textbf{Related work:} One-counter automata can also be seen as
pushdown automata with one letter stack alphabet. Stochastic games and
MPDs generated by pushdown automata and stateless pushdown automata
(also known as BPA) with termination and reachability objectives have
been studied in \cite{EY:RMC-RMDP,EY:RMDP-efficient,%
  BBFK:BPA-games-reachability-IC,BBKO:BPA-games-reachability-IC}.  To
the best of our knowledge, the only prior work on the expected
termination time (or, more generally, total accumulated reward)
objective for a class of infinite-state MDPs or stochastic games is
\cite{EWY:RSG-Positive-Rewards}, where this problem is studied for
stochastic BPA games. The termination objective for one-counter MDPs
and games has been examined in
\cite{BBEKW:OC-MDP,BBE:OC-games,BBEK:OC-games-termination-approx},
where it was shown (among other things) that the equilibrium
termination probability (i.e., the termination value) can be
approximated up to a given precision in exponential time, but no
lower bound was provided. The games over one-counter automata are
also known as ``energy games''
\cite{CHD:energy-games,CHDHR:energy-mean-payoff}.  Intuitively, the
counter is used to model the amount of currently available energy, and
the aim of the controller is to optimize the energy
consumptions. Finally, let us note that OC-MDPs can be seen as
discrete-time Quasi-Birth-Death Processes (QBDs, see, e.g.,
\cite{LR:book,EWY:one-counter-PE}) extended with a control. Hence, the theory of
one-counter MDPs and games is closely related to queuing theory, where
QBDs are considered as a fundamental model.

\section{Preliminaries}
\label{sec-defs}
Given a set $A$, we use $|A|$ to denote the cardinality of $A$.  We
also write $|x|$ to denote the absolute value of a given $x \in \Rset$,
but this should not cause any confusions.  The encoding
size of a given object $B$ is denoted by $\size B$. The set of
integers is denoted by $\Zset$, and the set of positive integers 
by~$\Nset$.

We assume familiarity with  basic notions of probability theory.
In particular, we call a probability distribution
$f$ over a discrete set $A$ \emph{positive} if $f(a) > 0$ for all
$a \in A$, and \emph{Dirac} if $f(a) = 1$ for some $a \in A$.

\begin{definition}[MDP]
A \emph{Markov decision process (MDP)} is a tuple 
$\M = (S,(S_0,S_1),\btran{},\Prob)$,
consisting of a countable set of \emph{states} $S$
partitioned into the sets $S_0$ and $S_1$ of \emph{stochastic}
and \emph{non-deterministic} states, respectively.
The \emph{edge relation} ${\btran{}} \subseteq S \times S$
is total, i.e., for every $r\in S$ there is $s\in S$ such that 
$r \btran{} s$. Finally, $\Prob$ assigns to every $s \in S_0$
a positive probability distribution over its outgoing edges.
\end{definition}

\noindent
A \emph{finite path} is a sequence $w=s_0 s_1 \cdots s_n$ of states
such that $s_i \btran{} s_{i+1}$ for all $0\leq i < n$. We write
$\len{w}=n$ for the length of the path.  A \emph{run} is an
infinite sequence $\omega$ of states such that every finite prefix 
of~$\omega$ is a path.
For a finite path, $w$, we denote by $\runs{w}$ the set of runs having
$w$ as a prefix.  These generate the standard $\sigma$-algebra on the
set of runs.

\begin{definition}[OC-MDP]
  A \emph{one-counter MDP (OC-MDP)} is a tuple $\A =
  (Q,(Q_0,Q_1),\delta,P)$, where $Q$ is a finite non-empty set of
  \emph{control states} partitioned into stochastic and non-deterministic
  states (as in the case of MDPs), $\delta\subseteq
  Q\times\{+1,0,-1\}\times Q$ is a set of \emph{transition rules} such
  that $\delta(q)\coloneqq\{(q,i,r) \in \delta\}\neq\emptyset$ for all
  $q\in Q$, and $P=\{P_q\}_{q\in\stochS{Q}}$ where $P_q$ is a positive
  rational probability distribution over $\delta(q)$ for all $q\in
  Q_0$.
\end{definition}

\noindent
In the rest of this paper we often write $q \ltran{i} r$ to indicate
that $(q,i,r)\in\delta$, and $q \ltran{i,x} r$ to indicate that
$(q,i,r)\in\delta$, $q$ is stochastic, and $P_q(q,i,r) = x$.
Without restrictions, we assume that for each pair
$q,r\in Q$ there is at most one $i$ such that $(q,i,r)\in\delta$.
The encoding size of $\A$ is denoted by 
$\size\A$, where all numerical constants are encoded as
fractions of binary numbers. The set of all \emph{configurations}
is $\cf\A \coloneqq \{\conf{q}{i} \mid q\in Q, i\in \Zset\}$.

To $\A$ we associate an infinite-state MDP
$\M^{\infty}_{\A} = (\cf\A,(\cf\A_0,\cf\A_1),\btran{},\Prob)$, where
the partition of $\cf\A$ is defined by
$\conf{q}{i}\in\stochS{\cf{\A}}$ iff $q\in\stochS{Q}$,
and similarly for $\cf\A_1$. The edges are defined by
$\conf{q}{i}\btran{}\conf{r}{j}$ iff $(q,j-i,r)\in\delta$.
The probability assignment $Prob$ is derived naturally
from~$P$.

By forgetting the counter values,
the OC-MDP $\A$ also defines a finite-state MDP
$\M_\A = (Q,(Q_0,Q_1), \btran{},\Prob')$.
Here $q\btran{} r$ iff $(q,i,r)\in\delta$ for some $i$,
and $Prob'$ is derived in the obvious way from $P$
by forgetting the counter changes.
\smallskip

\noindent
\textbf{Strategies and Probability.}\quad Let $\M$ be an MDP.
A \emph{history} is a finite path in $\M$, and a \emph{strategy} 
(or \emph{policy}) is a function assigning to each history
ending in a state from $S_1$ a distribution on edges leaving the 
last state of the history.
A strategy $\sigma$ is \emph{pure} (or \emph{deterministic}) if it 
always assigns $1$ to one edge and $0$ to the others, and 
\emph{memoryless} if $\sigma(w) = \sigma(s)$ where $s$ is the last state
of a history $w$.

Now consider some OC-MDP $\A$. A strategy $\sigma$ over the histories 
in $\M^{\infty}_{\A}$ is \emph{counterless} if it is memoryless and 
$\sigma(\conf{q}{i})=\sigma(\conf{q}{j})$ for all $i,j$.
Observe that every strategy $\sigma$ for $\M^{\infty}_{\A}$ gives a unique 
strategy $\sigma'$ for $\M_{\A}$ which just forgets the counter 
values in the history and plays as $\sigma$. This correspondence is
bijective when restricted to memoryless strategies in $\M_{\A}$
and counterless strategies in $\M^{\infty}_{\A}$, and it is used 
implicitly throughout the paper.

Fixing a strategy $\sigma$ and an initial state $s$,
we obtain in a standard way a probability measure $\Pr{\sigma}{s}{\cdot}$
on the subspace of runs starting in $s$.
For MDPs of the form $\M^{\infty}_{\A}$ for some OC-MDP $\A$, we consider 
two sequences of random variables,
$\{\counter{i}\}_{i\geq0}$
and
$\{\cstate{i}\}_{i\geq0}$,
returning the current counter value and the current control state 
after completing $i$ transitions.
\smallskip

\noindent
\textbf{Termination Time in OC-MDPs.}\quad
Let $\A$ be a OC-MDP. A run $\omega$ in $\M^\infty_{\A}$
\emph{terminates} if $\omega(j) = q(0)$ for some $j \geq 0$ and $q \in
Q$. The associated \emph{termination time}, denoted by $T(\omega)$, is
the least $j$ such that $\omega(j) = q(0)$ for some $q \in Q$. If
there is no such~$j$, we put $T(\omega) = \infty$, where the symbol
$\infty$ denotes the ``infinite amount'' with the standard
conventions, i.e., \mbox{$c < \infty$} and $\infty + c = \infty +
\infty = \infty \cdot d = \infty$ for arbitrary real numbers~$c,d$
where $d > 0$. %

For every strategy
$\sigma$ and a configuration $\conf{q}{i}$, we use $\E^\sigma
\conf{q}{i}$ to denote the expected value of $T$ in the probability
space of all runs initiated in $\conf{q}{i}$ where
$\Pr{\sigma}{\conf{q}{i}}{\cdot}$ is the underlying probability
measure.  The \emph{value} of a given configuration $\conf{q}{i}$ is
defined by $\val(\conf{q}{i})\coloneqq \inf_{\sigma} \E^\sigma
\conf{q}{i}$.  Let $\varepsilon \geq 0$ and $i \geq 1$. 
We say that a constant $\nu$
approximates $\val(\conf{q}{i})$ up to the absolute or relative
error $\eps$ if $|\val(\conf{q}{i}) - \nu| \leq \eps$ or
$|\val(\conf{q}{i}) - \nu|/\val(\conf{q}{i}) \leq \eps$, respectively.
Note that if $\nu$ approximates $\val(\conf{q}{i})$ up to the absolute
error $\eps$, then it also approximates $\val(\conf{q}{i})$ up to
the relative error $\eps$ because $\val(\conf{q}{i}) \geq 1$.
A strategy $\sigma$ is (absolutely or relatively)
\emph{$\varepsilon$-optimal} if $\E^\sigma \conf{q}{i}$ approximates
$\val(\conf{q}{i})$ up to the (absolute or relative) error~$\eps$.
A $0$-optimal strategy is called \emph{optimal}.

It is easy to see that there is a memoryless deterministic strategy
$\sigma$ in $\M^\infty_{\A}$ which is optimal in every configuration
of $\M^\infty_{\A}$.
First, observe that for all $q \in Q_0$, $q' \in Q_1$, 
and $i \neq 0$ we have that
\[
\begin{array}{lcl}
   \val(\conf{q}{i})  & = & 1 + \sum_{\conf{q}{i} \btran{x} \conf{r}{j}} x \cdot 
          \val(\conf{r}{j})\\[1ex]
   \val(\conf{q'}{i}) & = & 1 + \min\{\val(\conf{r}{j}) \mid \conf{q'}{i} 
          \btran{} \conf{r}{j}\}.
\end{array}
\]
We put $\sigma(\conf{q}{i}) = \conf{r}{j}$ where 
$\conf{q}{i} \btran{} \conf{r}{j}$ and 
$\val(\conf{q}{i}) = 1 + \conf{r}{j}$ (if there are several candidates for
$\conf{r}{j}$, any of them can be chosen). Now we can easily verify 
that $\sigma$ is indeed optimal in every configuration.

\section{Upper Bounds}
\label{sec-upper}

The goal of this section is to prove the following:
\begin{theorem}\label{thm:upper-main}\ 
  Let  $\A$ be a OC-MDP, $\conf{q}{i}$ a configuration of $\A$ where
  $i \geq 0$, and $\eps >0$.%
  \begin{enumerate}
  \item The problem whether $\val(\conf{q}{i})=\infty$ is decidable in 
     polynomial time.
  \item There is an algorithm that computes a rational number $\nu$ 
    such that \mbox{$|\val(\conf{q}{i})-\nu|\leq\eps$}, and a strategy $\sigma$ 
    that is absolutely $\eps$-optimal starting in $\conf{q}{i}$.
    The algorithm  runs in time exponential in $\size\A$ and polynomial 
    in $i$ and $1/\varepsilon$. (Note that $\nu$ then approximates 
    $\val(\conf{q}{i})$ also up to the relative error~$\eps$, and 
    $\sigma$ is also relatively $\eps$-optimal in $\conf{q}{i}$).
\end{enumerate}
\end{theorem}
For the rest of this section, we fix an OC-MDP $\A =
(Q,(Q_0,Q_1),\delta,P)$. First, we prove Theorem~\ref{thm:upper-main}
under the assumption that $\M_\A$ is \emph{strongly connected}
(Section~\ref{sec:scMDP}).  A generalization to arbitrary OC-MDP is
then given in Section~\ref{sec:nscMDP}.

\subsection{Strongly connected OC-MDP}\label{sec:scMDP}

\begin{figure}[t]
\begin{align*}
\makebox[2em]{~~~~~~~~\textit{maximize} $x$, \textrm{ subject to }}\\
\poten_q
&\leq
-x
+
k
+
\poten_r
&
\text{for all $q\in\pOneS{Q}$ and $(q,k,r)\in\delta$,}
\\
\poten_q
&\leq
-x
+
\textstyle\sum_{(q,k,r)\in\delta}
P_q((q,k,r))\cdot
(k+\poten_r)
&
\text{for all $q\in\stochS{Q}$,}%
\end{align*}
\caption{The linear program $\sL$ over $x$ and $\poten_q$, $q\in Q$.}
\label{fig:system-L}
\end{figure}

Let us assume that $\M_\A$ is strongly connected, i.e., for all $p,q
\in Q$ there is a finite path from $p$ to $q$ in~$\M_\A$.  Consider
the linear program of~Figure~\ref{fig:system-L}. Intuitively, the
variable $x$ encodes a lower bound on the long-run trend of the
counter value.  More precisely, the maximal value of~$x$ corresponds
to the \emph{minimal} long-run average change in the counter value
achievable by some strategy. The program corresponds to the one used
for optimizing the long-run average reward in Sections~8.8 and~9.5
of~\cite{Puterman:book}, and hence we know it has a solution.

\begin{lemma}[\cite{Puterman:book}]
\label{lem:sol}
  There is a rational solution
  $\left(\bar{x},(\bar{\poten}_q)_{q\in Q}\right)\in \Qset^{|Q|+1}$ to
  $\sL$, and the encoding size\footnote{Recall that rational numbers 
  are represented as fractions of binary numbers.}
  of the solution is polynomial in
  $\size{\A}$.
\end{lemma}
Note that $\bar{x}\geq -1$, because for any fixed $x \leq -1$ the
program $\sL$ trivially has a feasible solution.  
Further,  we put $V := \max_{q \in
  Q}{\bar{\poten}_q} - \min_{q \in Q}{\bar{\poten}_q}$. 
Observe that $V \in\exp\left(\size{\A}^{\calO(1)} \right)$ and $V$ is
computable in time polynominal in $\size{\A}$.
\begin{proposition}\label{prop:charact}
Let $\left(\bar{x},(\bar{\poten}_q)_{q\in Q}\right)$ be a solution
of~$\sL$. 
\begin{enumerate}
\item[(A)] If $\bar{x}\geq 0$, then $\val(\conf{q}{i})=\infty$ for
   all $q \in Q$ and $i \geq |Q|$.
\item[(B)] If $\bar{x}<0$, then the following holds:
  \begin{itemize}
    \item[(B.1)] For \emph{every} strategy $\pi$ and all $q \in Q$, $i \geq 0$
      we have that 
      $
        \E^\pi\conf{q}{i} \ \geq \ 
        (i-V)/|\bar{x}|.
      $
    \item[(B.2)] There is a counterless strategy $\sigma$ and a number 
      $U\in\exp\left(\size{\A}^{\calO(1)} \right)$ such that for
      all $q \in Q$, $i \geq 0$  we have that    
      $
        \E^\sigma\conf{q}{i} \leq (i+U)/|\bar{x}|.
      $
      Moreover, $\sigma$ and $U$ are computable in time polynomial 
      in $\size{\A}$.
  \end{itemize}
\end{enumerate}
\end{proposition}
First, let us realize that Proposition~\ref{prop:charact} 
implies Theorem~\ref{thm:upper-main}. To see this, we
consider the cases $\bar{x}\geq 0$ and $\bar{x}<0$ separately.
In both cases, we resort to analyzing a finite-state MDP  
$\mathcal{G}_K$, where $K$ is a suitable natural number, obtained
by restricting $\M^{\infty}_{\A}$ to configurations with counter value at
most $K$, and by substituting all transitions leaving each $p(K)$ with
a self-loop of the form $p(K) \btran{} p(K)$.

First, let us assume that $\bar{x}\geq 0$. By 
Proposition~\ref{prop:charact}~(A), we have that 
$\val(\conf{q}{i})=\infty$ for all $q\in Q$ and $i\geq |Q|$. Hence, it 
remains to approximate the value and compute \mbox{$\varepsilon$-optimal}
strategy for all configurations $\conf{q}{i}$ where $i \leq |Q|$.
Actually, we can even compute these values precisely and construct a 
strategy $\hat{\sigma}$ which is optimal in each such $\conf{q}{i}$.
This is achieved simply by considering the finite-state MDP $\mathcal{G}_{|Q|}$
and solving the objective of minimizing the expected number of transitions
needed to reach a state of the form $\conf{p}{0}$, which can be done 
by standard methods in time polynomial in $\size\A$.

If $\bar{x} < 0$, we argue as follows. 
The strategy $\sigma$
of Proposition~\ref{prop:charact}~(B.2) is not necessarily 
$\eps$-optimal in~$\conf{q}{i}$, so we cannot use it directly. 
To overcome this problem,
consider an \emph{optimal} strategy $\pi^*$ in~$\conf{q}{i}$, and let
$x_\ell$ be the probability that a run initiated in~$\conf{q}{i}$ (under
the strategy~$\pi^*$) visits a configuration of the form $\conf{r}{i+\ell}$.
Obviously,  $x_\ell \cdot \min_{r \in Q} \{\E^{\pi^*} \conf{r}{i{+}\ell}\}
\leq \E^\sigma \conf{q}{i}$, because otherwise $\pi^*$ would not be
optimal in $\conf{q}{i}$. Using the lower/upper bounds for 
$\E^{\pi^*} \conf{r}{i{+}\ell}$ and $\E^\sigma \conf{q}{i}$ given
in Proposition~\ref{prop:charact}~(B), we obtain
$x_\ell \leq (i+U)/(i+\ell-V)$. Then, we compute $k \in \Nset$ such that 
\[
   x_k \cdot \left(\max_{r \in Q} \left\{(i+k+U)/|\bar{x}| - 
              \E^{\pi^*} \conf{r}{i{+}k}\right\}\right) 
   \quad \leq \quad \eps
\]
A simple computation reveals that it suffices to put 
\[
   k \quad \geq \quad \frac{(i+U)(U+V)}{\eps |\bar{x}|} + V -i
\]
Now, consider $\mathcal{G}_{i+k}$, and let $f$ be a reward function 
over the transitions of $\mathcal{G}_{i+k}$ such that the loops on
configurations where the counter equals $0$ or $i+k$ have zero
reward, a transition leading to a state $\conf{r}{i{+}k}$ has 
reward $(i+k+U)/|\bar{x}|$, and all of the remaining transitions 
have reward~$1$.
Now we solve the finite-state MDP $\mathcal{G}_{i+k}$ with the 
objective of minimizing the total accumulated reward. 
Note that an optimal strategy 
$\varrho$ in $\mathcal{G}_{i+k}$ is computable in time polynomial
in the size of $\mathcal{G}_{i+k}$ \cite{Puterman:book}. 
Then, we define the corresponding
strategy $\hat{\sigma}$ in $\M^{\infty}_{\A}$, which behaves like 
$\varrho$ until the counter reaches $i+k$, and from that point on
it behaves like the counterless strategy~$\sigma$. It is easy to see
that $\hat{\sigma}$ is indeed $\eps$-optimal in~$\conf{q}{i}$.
\smallskip

\noindent
\textbf{Proof of Proposition~\ref{prop:charact}.}
Similarly as in \cite{BBEK:OC-games-termination-approx}, 
we use the solution $(\bar{x},(\bar{\poten}_q)_{q\in Q})\in \Qset^{|Q|+1}$ of 
$\sL$ to define a suitable submartingale, which is then used to derive
the required bounds. In \cite{BBEK:OC-games-termination-approx}, 
Azuma's inequality was applied to the submartingale to prove exponential tail 
bounds for termination probability. In this paper, we need to use
the optional stopping theorem rather than Azuma's inequality, and
therefore we need to define the submartingale relative to a suitable 
filtration so that we can introduce an appropriate stopping time (without the
filtration, the stopping time would have to depend just on numerical
values returned by the martingale, which does not suit our purposes). 

Recall the random variables $\{\counter{i}\}_{i\geq0}$ and
$\{\cstate{i}\}_{i\geq0}$ 
returning the height of the counter, and the control state
after completing $i$ transitions, respectively.
Given the solution
$(\bar{x},(\bar{\poten}_q)_{q\in Q})\in \Qset^{|Q|+1}$
from Lemma~\ref{lem:sol},
we define a sequence of random variables
$\{\mar{i}\}_{i\geq0}$
by setting
\[
\mar{i}
\coloneqq
\begin{cases}
\counter{i} +
\bar{\poten}_{\cstate{i}} - i\cdot \bar{x}
&
\text{if $\counter{j}>0$ for all $j,\ 0\leq j< i$,}\\
\mar{i-1} & \text{otherwise.}
\end{cases}
\]
Note that for every history $u$ of length $i$ and every $0 \leq j \leq i$,
the random variable $\mar{j}$ returns the same value for every
$\omega \in \runs{u}$. The same holds for variables $\cstate{j}$ and
$\counter{j}$. We will denote these common values $\mar{j}(u)$,
$\cstate{j}(u)$ and $\counter{j}(u)$, respectively.  Using the same
arguments as in~Lemma~3 of~\cite{BBEK:OC-games-termination-approx},
one may show that for every history $u$ of length $i$ we
have %
$\E(\mar{i+1}\mid \runs{u}) \geq \mar{i}(u)$.
This shows that $\{\mar{i}\}_{i\geq0}$ is a \emph{submartingale
  relative to the filtration} $\{\mathcal{F}_{i}\}_{i \geq 0}$, where
for each $i \geq 0$ the \mbox{$\sigma$-algebra} $\mathcal{F}_i$ is the 
\mbox{$\sigma$-algebra} generated by all $\runs{u}$ where
$\len{u}=i$. Intuitively, this means that value $\mar{i}(\omega)$ is
uniquely determined by prefix of $\omega$ of length $i$ and that the
process $\{\mar{i}\}_{i\geq0}$ has nonnegative average change.  For
relevant definitions of (sub)martingales see, e.g., \cite{Williams:book}.
Another important observation is that $|\mar{i+1}-\mar{i}| \leq
1+\bar{\poten}+V$ for every $i \geq 0$, i.e., the differences of the
submartingale are bounded.
\begin{lemma}
\label{lem:submart}
  Under an arbitrary strategy $\tau$ and with an
  arbitrary initial configuration $\conf{q}{j}$ where $j \geq 0$, the process
  $\{\mar{i}\}_{i\geq0}$ is a submartingale (relative to the
  filtration $\{\mathcal{F}_i\}_{i \geq 0}$) with bounded differences.
\end{lemma}

\paragraph{Part (A) of Proposition~\ref{prop:charact}.} 
This part can be proved 
by a routine application of the
optional stopping theorem to the martingale $\{\mar{i}\}_{i\geq0}$.
Let $\bar{z}_{\max}\coloneqq \max_{q\in Q} \bar{z}_q$, and
consider a configuration $p(\ell)$ where
$\ell+\bar{z}_r>\bar{z}_{\max}$. Let $\sigma$ be a strategy which is
optimal in every configuration. Assume, for the sake of contradiction,
that $\val(p(\ell))<\infty$. 

Let us fix $k\in \Nset$ such that $\ell+\bar{z}_r<\bar{z}_{\max}+k$
and define a stopping time $\tau$ which returns the first point in
time in which either $\mar{\tau}\geq \bar{z}_{\max}+k$, or
$\mar{\tau}\leq \bar{z}_{\max}$. To apply the optional stopping
theorem, we need to show that the expectation of $\tau$ is finite.

We argue that every configuration $q(i)$ with $i\geq 1$ satisfies the
following: under the optimal strategy $\sigma$, a configuration
with counter height $i-1$ is reachable from $q(i)$ in at most $|Q|^2$ 
steps (i.e., with a bounded probability). To see this, realize
that for every configuration $r(j)$ there is a successor, say $r'(j')$, 
such that $\val(r(j))>\val(r'(j'))$. Now consider a run $w$
initiated in $q(i)$ obtained by subsequently choosing successors 
with smaller and smaller values. Note that whenever $w(j)$ and
$w(j')$ with $j<j'$ have the same control state, the counter height of
$w(j')$ must be strictly smaller than the one of $w(j)$ because
otherwise the strategy $\sigma$ could be improved (it suffices to
behave in $w(j)$ as in $w(j')$). It follows that there
must be $k\leq |Q|^2$ such that the counter height of $w(k)$ is $i-1$.
From this we obtain that the expected value of $\tau$ is finite because the
probability of terminating from any configuration with bounded counter
height is bounded from zero. 
Now we apply the optional stopping
theorem and obtain $\Prb{\sigma}{p(\ell)}(\mar{\tau}\geq
\bar{z}_{\max}+k)\geq c/(k+d)$ for suitable constants
$c,d>0$. As $\mar{\tau}\geq \bar{z}_{\max}+k$ implies
$\counter{\tau}\geq k$, we obtain that
\[
\Prb{\sigma}{p(\ell)}(T\geq k)\quad\geq\quad
\Prb{\sigma}{p(\ell)}(\counter{\tau}\geq k)\quad\geq\quad
\Prb{\sigma}{p(\ell)}(\mar{\tau}\geq \bar{z}_{\max}+k)\quad\geq\quad
\frac{c}{k+d}
\]
and thus
\[
\E^\sigma \conf{p}{\ell}\quad = \quad \sum_{k=1}^{\infty}
\Prb{\sigma}{p(\ell)}(T\geq k)\quad \geq \quad \sum_{k=1}^{\infty}
\frac{c}{k+d}\quad = \quad \infty
\]
which contradicts our assumption that $\sigma$ is optimal and
$\val(p(\ell))<\infty$. 

It remains to show that $\val(p(\ell)) = \infty$ 
even for $\ell = |Q|$. This follows from the following simple observation:

\begin{lemma}
\label{lem:infty-low-boundary}
For all $q \in Q$ and $i \geq |Q|$ we have that 
$\val(\conf{q}{i}) < \infty$ iff $\val(\conf{q}{|Q|})< \infty$. 
\end{lemma}

\noindent
The ``only if'' direction of Lemma~\ref{lem:infty-low-boundary} is trivial.
For the other direction, let  $\mathcal{B}_k$ denote
the set of all $p\in Q$ such that $\val(\conf{p}{k})<\infty$.
Clearly, $\mathcal{B}_0=Q$, $\mathcal{B}_{k}\subseteq \mathcal{B}_{k-1}$, 
and one can easily verify that $\mathcal{B}_{k}=\mathcal{B}_{k+1}$
implies $\mathcal{B}_{k}=\mathcal{B}_{k+\ell}$ for all $\ell\geq 0$. Hence,
$\mathcal{B}_{|Q|} = \mathcal{B}_{|Q|+\ell}$ for all $\ell$. Note that
Lemma~\ref{lem:infty-low-boundary} holds for general OC-MDPs (i.e., 
we do not need to assume that $\M_\A$ is strongly connected).

\paragraph{Part (B1) of Proposition~\ref{prop:charact}.}
Let $\pi$ be a strategy and $\conf{q}{i}$ a configuration where 
$i \geq 0$. If $\E^{\pi}\conf{q}{i} = \infty$, we are done.
Now assume $\E^{\pi}\conf{q}{i} < \infty$.
Observe that for every $k \geq 0$ and every run~$\omega$, the
membership of $\omega$ into $\{T \leq k\}$ depends only on the finite 
prefix of $\omega$ of length~$k$. This means that $T$ is a stopping time
relative to filtration $\{\mathcal{F}_{n}\}_{n \geq 0}$. Since 
$\E^{\pi}\conf{q}{i} < \infty$ and the submartingale 
$\{\mar{n}\}_{n \geq 0}$ has bounded differences, we
can apply the optional stopping theorem and obtain $\E^{\pi}(\mar{0})\leq
\E^{\pi}(\mar{T})$. But $\E^{\pi}(\mar{0})=i+\bar{z}_{q}$ and
$\E^{\pi}(\mar{T}) = \E^{\pi}\bar{\poten}_{\cstate{T}} +
\E^{\pi}\conf{q}{i} \cdot |\bar{x}|$. Thus, we get $\E^{\pi}\conf{q}{i}
\geq (i + \bar{z}_{q}- \E^{\pi}\bar{\poten}_{\cstate{T}})/|\bar{x}|
\geq (i-V)/|\bar{x}|$.

\paragraph{Part (B2) of Proposition~\ref{prop:charact}.}
First we show how to construct the desired strategy~$\sigma$.  
Recall again the linear program $\sL$ of Figure~\ref{fig:system-L}. 
We have already shown that this program has an
optimal solution $\left(\bar{x},(\bar{\poten}_q)_{q\in Q}\right)\in
\Qset^{|Q|+1}$, and we assume that $\bar{x}<0$. By the strong duality
theorem, this means that the linear program dual to $\sL$ also has a
feasible solution $\dualsol$. Let
\[
 D = \{ q \in Q_0 \mid \bar{\lpfreq}_q > 0 \} \cup \{ q \in Q_1 \mid \bar{\lpfreq}_{(q,i,q')} > 0 \text{ for some } (q,i,q') \in \delta \}.
\]
By Corollary~8.8.8 of \cite{Puterman:book}, the solution 
$\dualsol$ can be chosen so that for every $q \in Q_1$ there is at 
most one transition $(q,i,q')$ with $\bar{\lpfreq}_{(q,i,q')}>0$.
Following the construction given in Section~8.8 of~\cite{Puterman:book},
we define a counterless deterministic strategy $\sigma$ such that 
\begin{itemize}
 \item in a state $q \in D \cap Q_1$, the strategy $\sigma$ selects
    the transition $(q,i,q')$ with $\bar{\lpfreq}_{(q,i,q')}>0$;
 \item in the states outside $D$, the strategy $\sigma$  behaves like an
   optimal strategy for the objective of reaching the set~$D$.
\end{itemize}
Clearly, the strategy $\sigma$ is computable in time polynomial in $\size{\A}$.
To show that $\sigma$ indeed satisfies Part~(B.2) of 
Proposition~\ref{prop:charact}, we need to prove a series of auxiliary 
inequalities, which can be found in Appendix~\ref{app-B2}.

\subsection{General OC-MDP}
\label{sec:nscMDP}

In this section we prove Theorem~\ref{thm:upper-main} for general
OC-MDPs, i.e., we drop the assumption that $\M_{\A}$ is strongly
connected.  We say that $C\subseteq Q$ is an \emph{end component of
  $\A$} if $C$ is strongly connected and for every $p\in C\cap Q_0$ we
have that $\{q\in Q \mid p\btran{} q\} \subseteq C$. A \emph{maximal
  end component (MEC) of $\A$} is an end component of~$\A$ which is
maximal w.r.t.{} set inclusion.  The set of all MECs of $\A$ is
denoted by $\MEC(\A)$. Every $C \in \MEC(\A)$ determines a strongly
connected OC-MDP $\A_C=(C,(C\cap Q_0,C\cap Q_1),\delta\cap (C\times
\{+1,0,-1\}\times C), \{P_q\}_{q\in C\cap Q_0})$. Hence, we may apply
Proposition~\ref{prop:charact} to $\A_C$, and we use $\bar{x}_C$ and
$V_C$ to denote the constants of Proposition~\ref{prop:charact}
computed for~$\A_C$.

\paragraph{Part 1. of Theorem~\ref{thm:upper-main}.}
We show how to compute, in time polynomial in $\size\A$, 
the set $Q_{\fin} = \{p \in Q \mid \val(p(k))<\infty 
\text{ for all } k \geq 0\}$. From this we easily obtain
Part~1. of Theorem~\ref{thm:upper-main}, because for every 
configuration $\conf{q}{i}$ where $i \geq 0$ we have the following:
\begin{itemize}
\item if $i \geq |Q|$, then $\val(\conf{q}{i}) < \infty$ iff
   $q \in Q_{\fin}$ (see  Lemma~\ref{lem:infty-low-boundary});
\item if $i < |Q|$, then $\val(\conf{q}{i}) < \infty$ iff 
   the set 
  $\{\conf{p}{0}\mid p \in Q\} \cup \{\conf{p}{|Q|}\mid p \in Q_{\fin}\}$ 
  can be reached from $\conf{q}{i}$ with probability~$1$
  in the finite-state MDP $\mathcal{G}_{|Q|}$ defined in 
  Section~\ref{sec:scMDP} (here we again use 
  Lemma~\ref{lem:infty-low-boundary}).
\end{itemize}
So, it suffices to show how to compute the set $Q_{\fin}$ in polynomial
time. 

\begin{proposition}\label{prop:char-fin}
  Let $Q_{<0}$ be the set of all states from which the set \mbox{$H = \{q\in
  Q\mid q\text{ belongs to a MEC }C\text{ satisfying }\bar{x}_C<0\}$}
  is reachable with probability~$1$.  Then \mbox{$Q_{\fin} =
  Q_{<0}$}. Moreover, the membership to $Q_{<0}$ is decidable in time
  polynomial in~$\size\A$.
\end{proposition}

\paragraph{Part 2. of Theorem~\ref{thm:upper-main}.} 

First, we generalize Part~(B) of Proposition~\ref{prop:charact}
into the following:
\begin{proposition}
  \label{prop:lower-main} For every $q \in Q_{\fin}$ there is a
  number $t_q$ computable in time polynomial in $\size{\A}$
  such that $-1 \leq t_q < 0$,  
  $1/|t_q| \in\exp\left(\size{\A}^{\calO(1)} \right)$, and the following
  holds:
\begin{itemize}
\item[(A)] There is a counterless strategy $\sigma$ and a number $U \in
  \exp(\size{\A}^{\calO(1)})$ such that for every configuration
  $\conf{q}{i}$ where $q\in Q_{\fin}$ and $i \geq 0$ we have that 
  $\E^{\sigma} \conf{q}{i} \leq i/|t_q|+U$. Moreover, both 
  $\sigma$ and $U$ are computable in time
  polynomial in $\size{\A}$.
\item[(B)] There is a number $L\in \exp(\size{\A}^{\calO(1)})$ such that
  for every strategy $\pi$ and every configuration
  $\conf{q}{i}$ where $i \geq |Q|$ we have that $\E^{\pi}\geq
  i/|t_q|-L$. Moreover, $L$ is computable in time polynomial in~$\size{\A}$.
\end{itemize}
\end{proposition}

\noindent
Once the Proposition \ref{prop:lower-main} is proved, we can compute
an $\varepsilon$-optimal strategy for an arbitrary configuration 
$\conf{q}{i}$ where $q \in Q_{\fin}$ and $i \geq |Q|$ in exactly the same 
way (and with the same complexity) as in the strongly connected case.
Actually, it can also be used to compute the approximate values and
$\varepsilon$-optimal strategies for configurations $\conf{q}{j}$
such that  $q\not\in Q_{\fin}$ or $1 \leq j < |Q|$. Observe that
\begin{itemize}
\item if $q \not \in Q_{\fin}$ and $j\geq |Q|$, the value is infinite 
  by Part~1;
\item otherwise, we construct the finite-state MDP $\G_{|Q|}$ 
  (see Section~\ref{sec:scMDP}) where the loops on
  configurations with counter value~$0$ have reward~$0$, the loops on
  configurations of the form $r(|Q|)$ have reward~$0$ or $1$,
  depending on whether $r \in Q_{fin}$ or not, transitions leading to
  $r(|Q|)$ where $r \in Q_{fin}$ are rewarded with some
  $\varepsilon$-approximation of $\val{(\conf{r}{|Q|})}$, and all
  other transitions have reward~$1$. The reward function can be
  computed in time exponential in $\size{\A}$ by
  Proposition~\ref{prop:lower-main}, and the minimal total accumulated
  reward from $\conf{q}{j}$ in $\G_{|Q|}$, which can be computed by
  standard algorithms, is an $\varepsilon$-approximation of
  $\val({\conf{q}{j}})$. The corresponding $\varepsilon$-optimal
  strategy can be computed in the obvious way.
\end{itemize}

\noindent
The proof of Propositions~\ref{prop:char-fin} and~\ref{prop:lower-main}
can be found in Appendices~\ref{app:proof:prop:char-fin} and~\ref{app:proof-lower-main}, respectively.

\section{Lower Bounds}
\label{sec-lower}

In this section, we show that approximating $\val(\conf{q}{i})$ is
computationally hard, even if $i=1$ and the edge
probabilities in the underlying OC-MDP are all equal to $1/2$.
More precisely, we prove the following:

\begin{theorem}
\label{thm-hard}
  The value of a given configuration $\conf{q}{1}$ cannot be 
  approximated up to a given absolute/relative error $\varepsilon > 0$ 
  unless P=NP, even if all outgoing edges of all stochastic control states 
  in the underlying OC-MDP have probability $1/2$.
\end{theorem}

The proof of Theorem~\ref{thm-hard} is split into two phases, which
are relatively independent. First, we show that given a propositional
formula $\varphi$, one can efficiently compute an OC-MDP $\A$, a configuration
$\conf{p}{K}$ of $\A$, and a number $N$ such that the value of $p(K)$
is either $N-1$ or $N$ depending on whether $\varphi$ is satisfiable
or not, respectively. The numbers $K$ and $N$ are exponential in 
$\size\varphi$,
which means that their encoding size is polynomial (we represent all 
numerical constants in binary). Here we use the technique of encoding
propositional assignments into counter values presented in 
\cite{Kucera:OC-FS-weak-bisimilarity-TCS},
but we also need to invent some specific gadgets to deal with our specific
objective. 
The first part already implies that approximating $\val(\conf{q}{i})$
is computationally hard. In the second phase, we show that the same 
holds also for configurations where the counter is initiated to~$1$. 
This is achieved by employing another gadget which just increases the 
counter to an 
exponentially high value with a sufficiently large probability. 
The two phases are elaborated in Lemma~\ref{lem-lower} and 
Lemma~\ref{lem-increase} which can be found in Appendix~\ref{app-lower}.

\bibliography{str-long,concur}
\newpage
\appendix
\section{Appendix}
First, let us fix some additional notation that will be used throughout
the whole appendix.

Given a random variable $X$ we denote $\E^{\pi}_{q(i)} X$ the expected value of $X$ computed under strategy $\pi$ from initial configuration $q(i)$. Since the initial configuration will be fixed in most of the proofs, we will usually omit the subscript and write only $\E^{\pi} X$.

Also, given a random variable $X$ and an event $A$,
we use $\E (X\mid A)$ to denote the conditional expectation of $X$ 
given the event $A$.

We also use $p_{\min}^{\A}$ to denote the minimal positive transition
probability in $\A$. We will usually omit the superscript $\A$ if $\A$
is clear from the context.

We say that (finite or infinite) path $u$ in OC-MDP $\A$ \emph{hits} 
or \emph{reaches} a set $D \subseteq Q$ if $\cstate{i}(u)\in D$ for 
some $i \leq \len{u}$. We say that $u$ \emph{evades} $D$ 
if it does not hit $D$.

\subsection{Proof of part (B.2) of Proposition \ref{prop:charact}.}
\label{app-B2}

First, denote $\A^{\sigma}$ the finite one-counter Markov chain that results from application of counterless strategy $\sigma$ on $\A$. That is, $\A^{\sigma}=(Q,(Q,\emptyset),\delta,P^{\sigma})$ where $P_{q}^{\sigma}(q,i,r)=P(q,i,r)$ for every stochastic state $q$ of $\A$, while for every non-deterministic state $q$ of $\A$ we have that $P_{q}^{\sigma}$ is a Dirac distribution that gives probability $1$ to transition selected by $\sigma(q)$.

Theorem 8.8.6 of \cite{Puterman:book} now guarantees that the set $D$
is exactly the set of all recurrent states in $\A^{\sigma}$. In particular, in $\A^{\sigma}$ there is no transition leaving $D$.

Now, let us recall a fundamental result from theory of linear programming,
the Complementary slackness theorem. In essence,
this theorem states that whenever we have a pair of solutions $u$ and
$v$ of the primal and dual linear program, respectively, then the
following equivalence holds: The $j$-th component of $v$ is positive
iff $u$ satisfies the $j$-th inequality of the primal linear program
as an equality. We can apply this on our pair of solutions $\primsol$,
$\dualsol$ to obtain the following system of linear equations:

\begin{align*}
\bar{\poten_q}
&=
-\bar{x}
+
k
+
\bar{\poten_r}
&
\text{whenever $q\in\pOneS{Q} \cap D$ and $\sigma$ selects $(q,k,r)$,}
\\
\bar{\poten_q}
&=
-\bar{x}
+
\textstyle\sum_{(q,k,r)\in\delta}
P_q((q,k,r))\cdot
(k+\bar{\poten_r})
&
\text{for all $q\in\stochS{Q} \cap D$.}
\end{align*}

With the help of these equations, we can easily prove the following lemma:

\begin{lemma}
\label{lem:martingale-lemma}
Under strategy $\sigma$, for any initial configuration $q(i)$ with $q\in D$ and for any history $u$ of length $i$ we have $\E^{\sigma}(\mar{i+1}\mid \runs{u}) = \mar{i}(u)$. That is, $\{\mar{i}\}_{i\geq0}$ is a martingale relative to the filtration $\{\mathcal{F}_{i}\}_{i \geq 0}$. 
\end{lemma}
 \begin{proof}
  The proof is the same as the proof of Lemma 22 in \cite{BKK:pOC-time-LTL-martingale}. \qed
 \end{proof}

From results of \cite{BKK:pOC-time-LTL-martingale} (where termination time of one-counter Markov chains was studied) it follows that under strategy $\sigma$ the expected termination time is finite from every initial configuration of the form $q(i)$ with $q \in D$. To be more specific, we can prove the following:

\begin{lemma} %
 For every initial configuration $\conf{q}{i}$ with $q \in D$ there are numbers $N \in \Nset$, $0<a<1$ such that for every $n \geq N$ we have $\Pr{\sigma}{q(i)}{T=n} \leq a^n $.
\end{lemma}
\begin{proof}
 The proof is the same as proof of Proposition 7 in \cite{BKK:pOC-time-LTL-martingale}. \qed 
\end{proof}

The finiteness of termination time easily follows because $$\E^{\sigma}q(i) = \sum_{k \in \Nset} k\cdot\Pr{\sigma}{q(i)}{T=k} \leq N + \sum_{k \geq N} k\cdot\Pr{\sigma}{q(i)}{T= k} \leq N + \sum_{k \in \Nset}k \cdot a^k < \infty.$$

\begin{corollary}
 For any initial configuration $\conf{q}{i}$ with $q \in D$ we have $\E^{\sigma}\conf{q}{i} < \infty$.
\end{corollary}

\begin{lemma}
\label{lem:recurrent-upper-bound}
 For any initial configuration $q(i)$ where $q \in D$ we have 
$
\E^{\sigma}\conf{q}{i} \leq (i+V)/|\bar{x}|.
$
\end{lemma}
\begin{proof}
As in proof of part (B.1) of Proposition \ref{prop:charact} we want to use the Optional stopping theorem to prove that $\E^{\sigma}\mar{0} = \E^{\sigma}\mar{T}$. We just need to verify that the assumptions of this theorem hold. We have already argued that $\{\mar{i}\}_{i\geq 0}$ is a martingale and $T$ is a stopping time relative to the same filtration $\{\mathcal{F}_i\}_{i \geq 0}$. We have also observed that $\{\mar{i}\}_{i\geq 0}$ has bounded differences. From the previous corollary we also now, that the expectation of stopping time $T$ is finite. Thus, the Optional stopping theorem applies and we indeed have $\E^{\sigma}\mar{0} = \E^{\sigma}\mar{T}$. But $\E^{\sigma}\mar{0} = i+\bar{z}_{q}$ and $\E^{\sigma}\mar{T}=\E^{\sigma}\bar{\poten}_{\cstate{T}} + |\bar{x}|\cdot\E^{\sigma}q(i)$. This gives us $\E^{\sigma}\conf{q}{i} = (i + \bar{z}_{q}- \E^{\sigma}\bar{\poten}_{\cstate{T}})/|\bar{x}| \leq (i+V)/|\bar{x}|$.
\qed
\end{proof}

To prove part (B.2) of Proposition \ref{prop:charact} it remains to prove the upper bound for arbitrary initial state. Intuitively, every state outside $D$ is transient in $\A^{\sigma}$ and thus under $\sigma$ we must reach $D$ ``quickly''. Once $D$ is reached, we can apply the bound from previous lemma.

\begin{lemma}
\label{lem:sc-upper-bound-technical-lemma}
Let $q(i)$ be any initial configuration. Denote $p:=\exp(-{p_{\min}}^{|Q|} / |Q|)$ where $p_{\min}$ is the minimal nonzero probability in $\A$. %
 Then we have 
\[
 \E^{\sigma} q(i) \leq  \frac{i + V + 2|Q| + \frac{4}{(1-p)^2}}{|\bar{x}|}.
\]
\end{lemma}
 Before we prove Lemma \ref{lem:sc-upper-bound-technical-lemma}, we should mention that the Lemma directly implies inequality in part (B.2) of Proposition \ref{prop:charact}. Indeed, the desired inequality holds for $U=V +  2|Q| + \frac{4}{(1-p)^2}$. The required asymptotic bound on $U$ is easy to check: we just need to recall, that for every real number $x\in[0,1]$ we have $1-\exp(-x) \geq x/2$ and thus $1/(1-p)^2 \leq 4|Q|^2/{p_{\min}^{2|Q|}}$. This also shows that $U$ is computable in time polynomial in $\size{\A}$. %

\begin{proof}[Proof of Lemma \ref{lem:sc-upper-bound-technical-lemma}]
We can write
\begin{equation*}
  \E^{\sigma}q(i) = \E^{\sigma}( T_1 + T_2),
\end{equation*}
where $T_{1}(\omega)=k$ iff $k$ is the first point in time when either $\counter{k}(\omega)=0$ or $\cstate{k}(\omega)\in D$; and where $T_2$ returns the termination time measured from the first time when $D$ was hit (formally we have $T_2(\omega)=-T_1(\omega) + T(\omega)$ if $T_1(\omega)<\infty$ and $T_2(\omega)=0$ otherwise). We will bound expectations of $T_1$ and $T_2$ separately.

 Let's start with $T_1$. Any run $\omega$ with $T_1(\omega) \geq k$ must either terminate before hitting $D$ but after at least $k$ steps; or it has to hit $D$ after at least $k$ steps. In both cases $\omega$ has to evade $D$ for at least $k-1$ steps. %
 From e.g. Lemma 23 of \cite{BKK:pOC-time-LTL-martingale} we know, that probability of evading $D$ for at least $k-1 \geq |Q|-1$ steps is at most $2p^k$. %
 We get

\begin{align}
\E^{\sigma} T_1 = \sum_{k=1}^{\infty} k \cdot \Pr{\sigma}{q(i)}{T_1 = k} = \sum_{k=1}^{\infty} \Pr{\sigma}{q(i)}{T_1 \geq k} \leq |Q| + \sum_{k=|Q|+1}^{\infty} \Pr{\sigma}{q(i)}{T_1 \geq k} \nonumber \\
\leq |Q| + \sum_{k=|Q|+1}^{\infty} 2p^{k} \leq |Q| + \sum_{k=0}^{\infty} 2p^{{k}} \leq |Q| + \frac{2}{1-p}.  \label{scub1}
\end{align}

Let us now concentrate on $T_2$. For any $l > 0$ we denote $D_{l}$ the set of all runs that terminate after hitting $D$ and have a counter value $l$ when they hit $D$ for the first time. (Formally, $\omega \in D_l$ iff $T_1(\omega)<\infty$, $\cstate{T_1}(\omega)\in D$ and $\counter{T_1}(\omega)=l$.) We also denote $D_0$ the set of all runs that reach a configuration with zero counter before or simultaneously with hitting $D$ for the first time.
Then we have 
\begin{equation}
\label{scub2}
 \E^{\sigma} T_2 = \sum_{l=0}^{\infty} \E^{\sigma}(T_2\mid D_{l})\cdot\Pr{\sigma}{q(i)}{D_{l}}.
\end{equation}
Note that by Lemma \ref{lem:recurrent-upper-bound} we have for every $l \in \Nset$
\[
 \E^{\sigma}(T_2\mid D_{l}) \leq \frac{l+V}{|\bar{x}|}
\]
Particularly for every $l \leq i+|Q|$ we have 
\begin{align} 
 \E^{\sigma}(T_2\mid D_{l}) &\leq  \frac{i  + V}{|\bar{x}|} +\frac{|Q|}{|\bar{x}|} \label{scub3}.
\end{align}

On the other hand, for $l \geq i + |Q|$ we have $\Pr{\sigma}{q(i)}{D_{l}} \leq 2p^{l-i}$, since no run in $D_l$ can hit $D$ in less than $l-i$ steps. Moreover, for $l \geq i + |Q|$ we can write
\begin{align}
 \E^{\sigma}(T_2\mid D_{l}) \leq \frac{i+(l-i) + V}{|\bar{x}|} = \frac{i+V}{|\bar{x}|} + \frac{(l-i)}{|\bar{x}|}. \label{scub9}
\end{align}

Plugging \eqref{scub3} and \eqref{scub9} into \eqref{scub2} we can compute 
\begin{align}
 \E^{\sigma} T_2 &\leq \frac{i + V}{|\bar{x}|} +\frac{|Q|}{|\bar{x}|} + \sum_{l \geq i + |Q|}^{\infty} \frac{2\cdot(l-i)\cdot p^{l-i}}{|\bar{x}|}  \leq \frac{i + V}{|\bar{x}|} +\frac{|Q|}{|\bar{x}|} + \frac{2}{|\bar{x}|\cdot(1-p)^2}. \label{scub10}
\end{align}

Putting \eqref{scub1} and \eqref{scub10} together we obtain
\[
 \E^{\sigma} \leq \frac{i + V}{|\bar{x}|}+\frac{|Q|}{|\bar{x}|} + \frac{4}{|\bar{x}|\cdot(1-p)^2}  +|Q| \leq \frac{i + V + 2|Q| + \frac{4}{(1-p)^2}}{|\bar{x}|}.
\]
\qed
\end{proof}

\subsection{Proof of Proposition~\ref{prop:char-fin}}\label{app:proof:prop:char-fin}

First, consider the membership problem for $Q_{<0}$. A decomposition
of $Q$ into maximal end components can be computed in polynomial time
using standard algorithms (see,~e.g.~\cite{Puterman:book}). By solving
the system $\mathcal{L}$ for individual MECs, we obtain the trends
$\bar{x}_C$ that in turn determine the set $H$. Finally, solving, in
polynomial time, the qualitative reachability of $H$ for every state
$q$ we obtain the set $Q_{<0}$.

It remains to prove that $Q_{\fin}=Q_{<0}$. We prove both inclusions separately.

\noindent
`$\supseteq$': Assume that $p\in Q_{<0}$.  First, observe that if $p$
belongs to a MEC $C$ satisfying $\bar{x}_C<0$ then, by
Proposition~\ref{prop:charact}, there is a counterless strategy which
stays in $C$ and terminates in finite expected time. In particular,
$\val(p(\ell))$ is finite and depends linearly on $\ell$.

Assume that a strategy $\sigma$ almost surely reaches $H$ from $p$ in
$\A_\M$.  As almost sure reachability is solved using memory-less
strategies in finite MDPs, we may assume that $\sigma$ is memory-less.
Denote by $\mathcal{H}$ the set of all configuration of the form
$q(\ell)$ where either $q\in H$, or $\ell=0$.  The strategy $\sigma$
induces a counter-less strategy $\sigma'$ in $\A^{\infty}_\M$ which
reaches $\mathcal{H}$ with probability one.  Moreover, using
$\sigma'$, $\mathcal{H}$ is reachable with a positive probability from
any configuration in at most $|Q|$ steps. This means that the expected
time to reach $\mathcal{H}$ is finite and the probability of reaching
a configuration of $\mathcal{H}$ with counter value at most $\ell$
before any other configuration of $\mathcal{H}$ is bounded by
$\frac{c}{d^\ell}$ for suitable constants $c,d>0$. As $\val(q(\ell))$
depends linearly on $\ell$ for every $q\in H$, we obtain that the
expected termination time for $p(k)$ is finite.

\noindent
`$\subseteq$': %
We proceed by contradiction. Assume that $Q_{fin}\smallsetminus
Q_{<0}\not = \emptyset$.
The following Lemma formalizes the crucial idea. %
\begin{lemma}\label{lem:comps}
  Assuming $Q_{fin}\smallsetminus Q_{<0}\not = \emptyset$, there is a
  MEC $C$ satisfying $C\subseteq Q_{fin}\smallsetminus Q_{<0}$ for
  which the following holds: if $s\btran{} t$ where $s\in C$ and
  $t\not\in C$, then $t\in Q\smallsetminus Q_{fin}$.
\end{lemma}
\begin{proof}
 First, we prove that if $Q_{fin}\smallsetminus Q_{<0}\not = \emptyset$, then it contains at least one MEC.
Assume, to the contrary, that all MECs contained in $Q_{fin}$ are also contained in $Q_{<0}$. We claim that then $Q_{fin}\subseteq Q_{<0}$.
Indeed, consider $p\in Q_{fin}\smallsetminus Q_{<0}$.
Note that starting in $p$, almost every run eventually reaches a MEC no matter what strategy is used. Moreover, there is a strategy which almost surely stays within $Q_{fin}$ forever starting in $p$. Using such a strategy, almost all runs initiated in $p$ reach MECs contained in $Q_{fin}$ and hence also in $Q_{<0}$. Thus, by definition of $Q_{<0}$, we have $p\in Q_{<0}$ which contradicts $p\in Q_{fin}\smallsetminus Q_{<0}$.

If there is a MEC $C\subseteq Q_{fin}\smallsetminus Q_{<0}$ such that no transition $s\btran{} t$ satisfies $s\in C$ and $t\not\in C$, then we are done. 
Assume, to obtain a contradiction, that for every MEC $C\subseteq Q_{fin}\smallsetminus Q_{<0}$ there is $s_C\btran{} t_C$ such that $s_C\in C$ but $t_C\in Q_{fin}\smallsetminus C$. Then for every $t_C$ there is a strategy which stays within $Q_{fin}$.
Let us consider a strategy $\pi$ that does the following:
\begin{itemize}
\item in all states of every MEC $C$ satisfying $C\subseteq Q_{fin}\smallsetminus Q_{<0}$, the strategy $\pi$ strives to reach $s_C$ with probability one
\item in each $s_C$, the strategy $\pi$ takes the transition $s_C\btran{} t_C$ with probability one
\item in states of $Q_{fin}$ that do not belong to any MEC, the strategy $\pi$ stays in $Q_{fin}$.
\end{itemize}
Note that we may safely assume that $\pi$ is memory-less. Consider the Markov chain $M^{\pi}$ induced by $\pi$ on states of $Q_{fin}$.
There are two possibilities. First, every bottom strongly connected component (BSCC) of $M^{\pi}$ contains a state of $Q_{<0}$. Then $Q_{<0}$ is reachable with probability one using $\pi$ from states of $Q_{fin}\smallsetminus Q_{<0}$, a contradiction with definition of $Q_{<0}$. Assume that there is at least one BSCC of $M_{\pi}$ which does not contain states of $Q_{<0}$. However, then the BSCC contains only states of $Q_{fin}\smallsetminus Q_{<0}$. Thus, by definition of $\pi$, the BSCC must contain at least two MECs, a contradition with the definition of MEC.
\qed
\end{proof}

Now let $\ell$ be a counter value such that for every $q\in
Q\smallsetminus Q_{fin}$ we have that $\val(q(\ell))=\infty$.  Let
$\sigma$ be a strategy and consider $p(\ell+|Q|)$ where $p\in C$. We
prove that $p$ cannot belong to $Q_{fin}$ which contradicts
$C\subseteq Q_{fin}\smallsetminus Q_{<0}$.

There are two cases. First, assume that using $\sigma$, a
configuration of the form $q(k)$, where $k\geq \ell$ and $q\in
Q\smallsetminus C$, is reachable via configurations with counter
values at least $\ell$ whose control states belong to $C$. Then by Lemma \ref{lem:comps}, $q\in Q\smallsetminus Q_{fin}$ and thus the expected
termination time from $q(\ell)$ is infinite. It follows that the
termination time from $p(\ell+|Q|)$ using $\sigma$ is infinite as
well. Assume that there is no such a path, i.e. that the only way how
to leave $C$ from $p(\ell+|Q|)$ using $\sigma$ is to decrease the
counter value below $\ell$. But then the expected termination time
from $p(\ell+|Q|)$ using $\sigma$ is at least as large as
$\val(p(|Q|))$ in $\A_C$, which is infinite by
Proposition~\ref{prop:charact} due to $\bar{x}_C\geq 0$. In both cases
we obtain that $p\not\in Q_{fin}$, a contradiction.

Note that the number $l$ mentioned above can be bounded from above by
$|Q|$ by Lemma \ref{lem:infty-low-boundary}.

\subsection{Proof of Proposition~\ref{prop:lower-main}}\label{app:proof-lower-main}

First we introduce some notation: for any run $\omega$ we denote $\inf(\omega)$ the set of states that are visited infinitely often by $\omega$. For any MEC $C$ we denote $M_C = \{\omega\mid \inf(\omega)\subseteq C\}$. It is well known that under arbitrary strategy $\pi$ we have $\Pr{\pi}{}{\bigcup_{C \in \MEC(\A)} M_C} = 1$, i.e. that $\inf(\omega)$ is almost surely contained in some MEC.

For any state $q$ denote $\Sigma_q^{<0}$ the set of all strategies $\sigma$ with the property that $\Pr{\sigma}{q}{\{\omega \mid \omega \in M_C, \bar{x}_C \geq 0\}}=0$. Note that by Proposition \ref{prop:char-fin} we have $\Sigma_{q}^{<0}\neq\emptyset$ for all $q \in Q_{fin}$.

Let us start with part (B) of Proposition \ref{prop:lower-main}.
We want to describe a counterless strategy $\sigma$ that terminates ``quickly`` from any configuration $q(j)$ with $q \in Q_{fin}$. Part (B2) of Proposition \ref{prop:charact} gives us for every MEC $C$ counterless strategy $\sigma_C$ such that for any initial configuration $q(i)$ with $q \in C$ we have $\E^{\sigma_C}\conf{q}{i} \leq (i+U_C)/|\bar{x}_C|$, for some number $U_C$. Main idea behind construction of $\sigma$ is to stitch these strategies together in appropriate way.

We argue that the following should hold: %
  First, strategy $\sigma$ should be in $\Sigma_q^{<0}$ for all states $q \in Q_{fin}$. Otherwise, the finite Markov chain $\A^{\sigma}$ induced by $\sigma$ on states of $\A$ would have some bottom strongly connected component (BSCC) contained in MEC $C$ with $\bar{x}_C \geq 0$. By part (A) of Proposition \ref{prop:charact} this would mean that $\E^{\sigma} q(j) = \infty$ for some $j$. %

Second, strategy $\sigma$ should minimize the long-run average number of steps needed to decrease the counter value by one. Note that since $\bar{x}_C$ represents the minimal long-run average change in counter value in MEC $C$, the number $|\bar{x}_C|^{-1}$ represents exactly the long-run average time needed to decrease the counter by 1 in $C$, provided that \mbox{$\bar{x}_C<0$}. Thus, strategy $\sigma$ should minimize the weighted sum $\sum_{C \in \MEC(\A)} \Pr{\sigma}{q(i)}{M_C}\cdot |\bar{x}_C|^{-1}$ for any initial configuration $q(i)$ with $q\in Q_{fin}$. Note that the objective of minimizing $\sum_{C \in \MEC(\A)} \Pr{\sigma}{q(i)}{M_C}\cdot |\bar{x}_C|^{-1}$ does not depend in any way on counter values so it suffices to show that the sum is minimized for some (unspecified) initial counter value $i$.
More formally, for every state $q \in Q_{fin}$ there is unique number $t_q<0$ such that $|t_q|^{-1} = \inf_{\pi \in \Sigma_{q}^{<0}}\sum_{C \in \MEC(\A)} \Pr{\pi}{q(i)}{M_C}\cdot |\bar{x}_C|^{-1}$ for all $i$. We call $t_q$ the \emph{minimal trend} achievable from $q$. Our goal is to find counterless deterministic strategy $\sigma \in \Sigma^{<0}_{q}$ such that $\sum_{C \in \MEC(\A)} \Pr{\sigma}{q(i)}{M_C}\cdot |\bar{x}_C|^{-1} = |t_q|^{-1}$, for every $q \in Q_{fin}$ and every $i$.

Denote \mbox{$\bar{x}_0 = \max\{\bar{x}_C \mid C \in \MEC(\A),\bar{x}_C < 0\}$}.
In order to compute strategy $\sigma$ and numbers $t_q$, we transform $\A$ into a new finite-state MDP with rewards $\A_R$ by ``forgetting'' counter changes in $\A$ and defining a reward function $R$ on transitions in $\A$ as follows:
\[ R(s \btran{} t) = \begin{cases}
                      \frac{1}{\bar{x}_C} & \text{ if } s,t \in C \text{ and } \bar{x}_C < 0 \\
                      \frac{x_0^{-1}-1}{p_{\min}^{|Q|}} & \text{otherwise,}%
                     \end{cases}
\]
It is clear that $\A_R$ can be constructed in time polynomial in $\size{\A}$.
\begin{claim}\label{claim:fast-strategy-general}
In $\A_R$ the maximal average reward achievable from state $q \in Q_{fin}$ is equal to $t_q^{-1}$. 
Moreover, there is a memoryless deterministic strategy $\sigma_{R}$ in $\A_{R}$ such that for every state $q \in Q_{fin}$ we have $\sigma_{R}\in \Sigma_{q}^{<0}$ and $\sum_{C \in \MEC(\A)} \Pr{\sigma_R}{q}{M_C}\cdot |\bar{x}_C|^{-1} = |t_q|^{-1}$.
\end{claim}
\begin{proof}
 The existence of optimal memoryless deterministic strategy $\sigma$ for maximization of average reward follows from standard results on MDPs (see \cite{Puterman:book}). It is obvious that for any $q\in Q_{fin}$ and any strategy $\pi \in \Sigma_q^{<0}$ the average reward obtained with strategy $\pi$ in $\A_R$ is equal to $\sum_{C \in \MEC(\A)}  \Pr{\pi}{q}{M_C}\cdot|\bar{x}_C|^{-1}$. It thus suffices to prove that $\sigma_R \in \Sigma_q^{<0}$ for every $q \in Q_{fin}$. Denote $M^{\sigma_R}$ the finite Markov chain induced by $\sigma_R$ on states of $\A$. Assume, for the sake of contradiction, that $\sigma_R \not\in \Sigma_q^{<0}$ for some $q \in Q_{fin}$. Then there must be a BSCC $B$ of $M^{\sigma_R}$ reachable from $q$ that is contained in some MEC $C$ with $\bar{x}_C \geq 0$. In $M_{\sigma_R}$ there must be a path of length at most $|Q|$ from $q$ to $B$, which means that under $\sigma_R$ the probability of runs that have average reward $\frac{x_0^{-1}-1}{p_{\min}^{|Q|}}$ is at least $p_{\min}^{|Q|}$. Since no run in $\A_R$ has average reward greater than $-1$, it follows that average reward achieved from $q$ with $\sigma_R$ is at most $x_0^{-1}-1 -(1-p_{\min}^{|Q|}) < x_{0}^{-1}.$ But this is contradiction with $\sigma_R$ maximizing the average reward, since $\Sigma_q^{<0}\neq \emptyset$ and every strategy from $\Sigma_q^{<0}$ yields average reward at least $x_{0}^{-1}$. \qed
\end{proof}

Strategy $\sigma_R$ can be computed in polynomial time with standard algorithms (see, e.g., \cite{Puterman:book}). We can now construct the desired counterless strategy $\sigma$ as follows: denote $\A^{\sigma_R}$ the finite Markov chain induced by $\sigma_R$ on states of $\A$. Note that every bottom strongly connected component of $\A^{\sigma_R}$ is contained in exactly one MEC $C(B)$ of $\A$. Strategy $\sigma$ behaves in the same way as $\sigma_R$ until some BSCC $B$ of $\A^{\sigma_R}$ is reached. Then $\sigma$ starts to behave as $\sigma_{C(B)}$. It is easy to see that $\sigma \in \Sigma_q^{<0}$ for all states $q \in Q_{fin}$.

Clearly, for every MEC $C$ and every initial configuration $q(i)$ with $q \in Q_{fin}$ we have $\Pr{\sigma_R}{q}{M_C} = \Pr{\sigma}{q(i)}{M_C}$ and thus also $\sum_{C \in \MEC(\A)} \Pr{\sigma}{q(i)}{M_C}\cdot |\bar{x}_C|^{-1} = |t_q|^{-1}$ for every $i$. Note that numbers $t_q$  satisfy all conditions mentioned in the initial part of Proposition~\ref{prop:lower-main}. Moreover, we can prove the following upper bound on expected termination time under $\sigma$:

\begin{proposition}
\label{prop:gen-fast-bound}
There is a number $U \in \exp\left(\size{\A} ^{\calO(1)}\right)$ that is computable in time polynomial in $\size{\A}$ such that for any initial configuration $\conf{q}{i}$ with $q \in Q_{fin}$ and $i \geq |Q|$ we have $$\E^{\sigma} q(i) \leq \frac {i}{|t_q|}+U.$$
\end{proposition}
\begin{proof}
 The proof closely follows the proof of Lemma \ref{lem:sc-upper-bound-technical-lemma}. However, there is a new obstacle in a presence of components with different trends.

Since the strategy $\sigma_{\fast}$ is memoryless, its application on $\A$ yields a finite one-counter Markov chain $\A^{\sigma}$.
Denote $D$ the union of its bottom strongly connected components.
We can now write
\begin{equation}
 \E^{\sigma_{\fast}}q(i) = \E^{\sigma_{\fast}}(T_1 + T_2),
\end{equation}
where again $T_{1}(\omega)=k$ iff $k$ is the first point in time when $\omega$ hits either $D$ or reaches a configuration with a zero counter and $T_2$ is a time to hit a configuration with a zero counter after hitting $D$ ($T_2$ returns zero if the run never terminates or terminates before hitting $D$).

We will bound expectations of $T_1$ and $T_2$ separately.

The bound on $\E^{\sigma} T_1$ can be computed in exactly the same way as in Lemma \ref{lem:sc-upper-bound-technical-lemma}. Thus we can conclude that
\begin{equation}
 \label{nub1}
\E^{\sigma} T_1 \leq |Q| + \frac{2}{1-p},
\end{equation}
where $p=\exp(-{p_{\min}}^{|Q|} / |Q|)$.

Now we bound the expectations of $T_2$.
Recall that for any $l > 0$ %
we denote $D_{l}$ the set of all runs that terminate after reaching $D$ and have a value counter value exactly $l$ when they hit $D$ for the first time (and we denote $D_0$ the set of all runs that terminate before hitting $D$ or hit $D$ with counter value exactly 0).
Also recall, that $M_C$ denotes the set of all runs $\omega$ with $\inf(\omega) \subseteq C$ and that under arbitrary strategy $\pi$ we have $\sum_{C \in \MEC(\A)} \Pr{\pi}{}{M_C} = 1$. Finally, denote $D_{l}^C = M_C \cap D_{l}$.

As discussed in section \ref{sec:nscMDP}, we can apply Proposition \ref{prop:charact} to every MEC $C$ of $\A$ separately. Especially, by construction of $\sigma$ the following holds: for every MEC $C$ that contains some BSCC of $\A^{\sigma}$, the Proposition \ref{prop:charact} gives us number $U_C\in \exp(\size{\A}^{\calO(1)})$ such that $\E^{\sigma}p(j) \leq (j+U_C)/|\bar{x}_C| $, for every $p \in C$ and $j\geq 0$. Set $$U' = \max\{U_C \mid C \in\MEC(\A),~C \text{ contains some BSCC of } \A^{\sigma}\}.$$ Clearly we still have $U' \in \exp(\size{\A}^{\calO(1)})$.

We have
\begin{equation}
\label{nub2}
 \E^{\sigma} T_2 = \sum_{C \in \MEC{(\A)}}\sum_{l=0}^{\infty} \E^{\sigma}(T_2\mid D_{l}^C)\cdot \Pr{\sigma}{q(i)}{D_{l}^C}.
\end{equation}

As in proof of Lemma \ref{lem:sc-upper-bound-technical-lemma}, we can easily show that for any MEC $C$ and any $l \leq i + |Q|$ we have
\begin{equation}
 \label{nub6}
 \E^{\sigma}(T_2\mid D_{l}^C) \leq  \frac{i + |Q|+U'}{|\bar{x}_C|}.
\end{equation}

For every $C$ and every $ l \geq i + |Q|$ we have 
\begin{equation}
\label{nub7}
\E^{\sigma}(T_2\mid D_{l}^C) \leq \frac{i+U'}{|\bar{x}_C|} +  \frac{(l-i)}{|\bar{x}_C|}
\end{equation}
and $\Pr{\sigma}{q(j)}{D_l} \leq 2p^{l-i}$ (the latter holds by Lemma 23 of \cite{BKK:pOC-time-LTL-martingale}).

Recall that we have denoted $\bar{x}_0 = \max\{\bar{x}_C \mid C \in \MEC(\A),\bar{x}_C < 0\}$. Putting \eqref{nub6} and \eqref{nub7} together we obtain for any fixed $C \in\MEC(\A)$
\begin{align*}
\sum_{l=0}^{\infty}\E^{\sigma}(T_2\mid D_{l}^C)\cdot\Pr{\sigma}{q(i)}{D_{l}^C} &\leq \frac{i + |Q|+U'}{|\bar{x}_C|}\cdot\Pr{\sigma}{q(i)}{M_C} + \sum_{l=i+|Q|}^{\infty}\frac{(l-i)\cdot\Pr{\sigma}{q(i)}{D_l^C}}{|\bar{x}_C|} \nonumber \\ &\leq \frac{i + |Q|+U'}{|\bar{x}_C|}\cdot\Pr{\sigma}{q(i)}{M_C} + \sum_{l=i+|Q|}^{\infty}\frac{(l-i)\cdot\Pr{\sigma}{q(i)}{D_l^C}}{|\bar{x}_0|}.
\end{align*}

Moreover, from the definition of strategy $\sigma$ we know that \mbox{$\sum_{\C \in\MEC{\A}} \Pr{\sigma}{q(i)}{M_C}\cdot |\bar{x}_C|^{-1} = |t_q|^{-1}$.}
We can use this and continue from \eqref{nub2} as follows:
\begin{align}
 \E^{\sigma} T_2 &\leq \sum_{C \in \MEC{(\A)}}\left(\frac{i + |Q|+U'}{|\bar{x}_C|}\cdot\Pr{\sigma}{q(i)}{M_C} + \sum_{l=i+|Q|}^{\infty}\frac{(l-i)\cdot\Pr{\sigma}{q(i)}{D_l^C}}{|\bar{x}_0|} \right) \nonumber \\
 &= \sum_{C \in \MEC{(\A)}}\left(\frac{i + |Q|+U'}{|\bar{x}_C|}\cdot\Pr{\sigma}{q(i)}{M_C}\right) + \sum_{l=0}^{\infty}\frac{(l-i)\cdot\sum_{C \in\MEC(\A)}\Pr{\sigma}{q(i)}{D_l^C}}{|\bar{x}_0|}\nonumber \\
 &= \frac{i +|Q|+U'}{|t_q|} + \sum_{l=0}^{\infty}\frac{(l-i)\cdot\overbrace{\Pr{\sigma}{q(i)}{D_l}}^{\leq 2p^{l-i}}}{|\bar{x}_0|} \leq \frac{i +|Q|+U'}{|t_q|} + \frac{2}{|\bar{x}_0|\cdot(1-p)^2}. \label{nub8}
\end{align}

Combining \eqref{nub1} and \eqref{nub2} we can conclude that
\[
 \E^{\sigma}q(i) \leq \frac{i }{|t_q|} + \frac{ 2|Q| + U'}{|t_q|} + \frac{4}{|\bar{x}_0|\cdot(1-p)^2} \leq \frac{i }{|t_q|} + \frac{ 2|Q| + U'}{|\bar{x}_0|} + \frac{4}{|\bar{x}_0|\cdot(1-p)^2}.
\]

The inequality in Proposition \ref{prop:gen-fast-bound} thus holds for $U=\frac{ 2|Q| + U'}{|\bar{x}_0|} + \frac{4}{|\bar{x}_0|\cdot(1-p)^2}$. The desired asymptotic bound is again easy to check.
\qed
\end{proof}

It remains to prove part (B) of Proposition \ref{prop:lower-main} (with numbers $t_q$ being the minimal trends achievable from $q$). %

The following Claim shows, that in order to prove Proposition \ref{prop:lower-main} (B) it suffices to prove its validity for strategies in $\Sigma_{q}^{<0}$, because termination value under some arbitrarily fixed strategy can be approximated up to some exponential error by termination value under suitable strategy from $\Sigma_{q}^{<0}$.

 \begin{claim}
 There is a number $K_1 \in \exp\left(\size{\A}^{\calO(1)} \right)$ that is computable in time polynomial in $\size{\A}$, with the following property: for every strategy $\pi$ and any initial configuration $\conf{q}{i}$ with $i \geq |Q|$ there is a strategy $\pi' \in \Sigma^{<0}_q$ such that $\E^{\pi}q(i) \geq \E^{\pi'}q(i) - K_1$.
\end{claim}
\begin{proof}
  Set $K_1 = (|Q|+U)/|\bar{x}_0|$, where $U$ is the constant from Proposition \ref{prop:gen-fast-bound}. Fix arbitrary strategy $\pi$. If $\E^{\pi}q(i) = \infty$, then the inequality clearly holds for any strategy $\pi' \in \Sigma_{q}^{<0}$. Otherwise, since $i\geq |Q|$, with $\pi$ we must almost surely reach a configuration of the form $p(|Q|)$. For every such reachable configuration we must have $p \in Q_{fin}$, since otherwise we would have $\E^{\pi}\conf{q}{i}=\infty$ by Lemma \ref{lem:infty-low-boundary}. Define new strategy $\pi'$ as follows: $\pi'$ behaves in the same way as $\pi$ until the configuration with counter height $|Q|$ is reached: then it starts to behave as strategy $\sigma$ from Proposition \ref{prop:lower-main}. Then clearly $\pi' \in \Sigma_q^{<0}$ and by Proposition \ref{prop:gen-fast-bound} the switch to strategy $\sigma$ in height $|Q|$ cannot delay the termination for more than $(|Q|+U)/|\bar{x}_0|$ steps. \qed
 \end{proof}

Under strategy $\pi \in \Sigma_{q}^{<0}$ we never reach state from $Q\setminus Q_{fin}$, if we start from $q$. We can thus safely remove all states from $Q\setminus Q_{fin}$, together with adjacent transitions, without influencing the behavior under strategies from $\Sigma_{q}^{<0}$. In the following we can without loss of generality assume that $Q=Q_{fin}$ and that all strategies are in $\Sigma_{q}^{<0}$, for every state $q$.

We will now finish the proof in two steps. First, we observe that there is only a small probability that the run revisits (i.e. leaves and then visits again) some MEC many times. Actually, this probability decays exponentially in number of revisits. We call a transition $r(j)\btran{}r'(j')$ in $\M_{\A}^{\infty}$ a \emph{switch} if there exists some MEC $C$ such that $|\{r,r'\}\cap C|=1$. For any run $\omega$ we denote $\sharp(\omega)$ the number of switches on $\omega$ and we set $\switch(\omega)=\sharp(\omega)+1$. That is, random variable $\switch$ counts the number of maximal time intervals in which $\omega$ either stays within a single MEC or outside any MEC.

\begin{lemma}
\label{lem:switches-bound}
 For every strategy $\pi$, every initial configuration $\conf{q}{i}$ and every $k \in \Nset$
\[
 \Pr{\pi}{\conf{q}{i}}{\switch = k} \leq 8\cdot |Q| \cdot c^{k},
\]
where $c=\exp\left(\frac{-p_{\min}^{|Q|}}{2|Q|} \right)$.
\end{lemma}
\begin{proof}
If $p_{\min}=1$, i.e. there are no (truly) stochastic states, then MECs are actually strongly connected components, $W(\omega)\leq 2\cdot |Q|$ for every run $\omega$, and the Lemma trivially holds. Otherwise, we have $p_{\min}\leq 1/2$.
We can use the following:

\begin{claim} %
Let $\A$ be arbitrary OC-MDP and let $C$ be a MEC of $\A$. Further, let $q \not\in C$ be any state that can be reached from $C$ with probability 1 (under some strategy). Then, under arbitrary strategy, the probability of reaching $C$ from any initial configuration of the form $q(i)$ is at most $1-p_{\min}^{|Q|}$.
\end{claim}
Let $\rho$ be the strategy that maximizes the probability of reaching $C$ from $q$ in $\A_{\cal{M}}$. From standard results on MDPs we may assume that $\rho$ is memoryless. Denote $M^{\rho}$ the finite Markov chain induced by $\rho$ on states of $\A_{\cal{M}}$. There must be at least one BSCC $B$ of $M^{\rho}$ reachable from $q$ such that in $\A_{M}$ the probability of reaching $C$ from any state of $B$ is less than 1 under any strategy (otherwise, there would be a strategy that almost surely reaches $C$ from $q$ -- a contradiction with $C$ being a MEC). In particular, sets $B$ and $C$ are disjoint. Thus, the probability of \emph{not reaching} $C$ from $q$ under $\rho$ is at least as large as probability of hitting $B$ in $M^{\rho}$. Since $\rho$ is memoryless, there is a run in $M^{\rho}$ that reaches $B$ in at most $|Q|$ steps. Thus, the probability of hitting $B$ is at least $p_{\min}^{|Q|}$. %
\vskip 0.3 cm
Let us now finish proof of the Lemma.
For any MEC $C$ and any $l \in \Nset$ denote $R_C^l$ the set of all runs that leave a MEC $C$ and then return to it for at least $l$ times. The claim shows that under any strategy $\pi$ we have $\Pr{\pi}{\conf{q}{i}}{R_C^l} \leq (1-p_{\min}^{|Q|})^l$. Now if $\switch(\omega)=k$ then $\omega$ must have revisited some MEC $C$ at least $\lfloor\frac{k}{2|Q|}-2\rfloor$ times, i.e. $\omega \in R_C^{\lfloor\frac{k}{2|Q|} -2\rfloor}$ for some MEC $C$. Thus $\Pr{\pi}{\conf{q}{i}}{\switch = k} \leq |Q|\cdot(1-p_{\min}^{|Q|})^{\lfloor\frac{k}{2|Q|}-2 \rfloor}$. Denote $\alpha = (1-p_{\min}^{|Q|})$.

We have $\lfloor\frac{k}{2|Q|}-2 \rfloor \leq \frac{k}{2|Q|}-3$ and thus $\Pr{\pi}{q(i)}{\switch = k} \leq |Q|\cdot \alpha^{\frac{k}{2|Q|}-3}=
|Q|\cdot \alpha^{\frac{k}{2|Q|}}/\alpha^3 $. Since $p_{\min}\leq 1/2$, we have $1/\alpha^3 \leq 8$. Moreover, from calculus we know that for any real number $x$ we have $1-x \leq \exp(-x)$. This gives us $\Pr{\pi}{\conf{q}{i}}{\switch = k} \leq 8\cdot|Q|\cdot \exp\left(-\frac{p_{\min}^{|Q|}}{2|Q|} \right)^k$, and the proof is finished. \qed
\end{proof}

The crucial idea behind the proof of Proposition \ref{prop:lower-main} (B) is now the following: whenever the system stays either in some MEC or outside any MEC for some period of time, we may approximate its behavior (up to some constant error) using the results of section \ref{sec:scMDP} and standard probabilistic computations, respectively. We show, that it is possible to use these approximations to approximate the behavior of the whole system. The error of this new approximation now depends on the average number of time intervals when run stays in some or outside any MEC. The following crucial proposition formalizes this idea.

\begin{proposition}
\label{prop:upper-general-slowbound}
There is a number $K \in \exp\left(\size{\A}^{\calO(1)} \right)$ that is computable in time polynomial in $\size{\A}$, such that
for every memoryless deterministic strategy $\pi$ and every initial configuration $\conf{q}{i}$ we have
\[
\E^{\pi}\conf{q}{i} \geq \frac{i}{|t_q|} - K\cdot \E^{\pi}W.
\]
\end{proposition}

Before we present the rather technical proof of Proposition \ref{prop:upper-general-slowbound}, let us make sure that it already implies Proposition \ref{prop:lower-main}.

Let $\conf{q}{i}$ be any initial configuration. Fix a memoryless deterministic strategy $\pi$ that minimizes the expected termination time from $\conf{q}{i}$. From Lemma \ref{lem:switches-bound} we have $\E^{\pi}\switch = \sum_{k=0}^{\infty}k\cdot \Pr{\pi}{\conf{q}{i}}{W=k} \leq 8\cdot |Q|\cdot \sum_{k=0}^{\infty}k\cdot c^k = \frac{8\cdot |Q|}{(1-c)^2}$. From calculus we now that for every $0 \leq x \leq 1$ it holds $1-\exp(-x)\geq x/2$ and thus we have $E^{\pi}\switch \leq \frac{32\cdot |Q|^2}{p_{\min}^{|Q|}}$. Denote this upper bound $K'$. Clearly $K' \in \exp\left(\size{\A}^{\calO(1)} \right)$ is computable in time polynomial in $\size{\A}$. By Proposition \ref{prop:upper-general-slowbound} we have
\[
 \val(\conf{q}{i}) = \E^{\pi}{\conf{q}{i}} \geq \frac{i}{|t_q|} - K\cdot K'.
\]
Since $K\cdot K' \in \exp\left(\size{\A}^{\calO(1)} \right)$, this proves Proposition \ref{prop:lower-main}, which is what we needed to finish the proof of Theorem \ref{thm:upper-main} in general case.

\subsection*{Proof of Proposition \ref{prop:upper-general-slowbound}}
\label{app-general-slowbound}

Recall, that in the following we assume $Q=Q_{fin}$.

First, we need to present some technical observations. %

The following lemma is a slight generalization of part (B1) of Proposition \ref{prop:charact}.
\begin{lemma}
\label{lem:general-mec-azuma}
Let $q(l)$ be any initial configuration such that $q\in C$, for some MEC $C$ of $\A$. Denote $\tleave$ the random variable that returns the first point in time when the run either terminates or reaches configuration of the form $r(j)$ with $r\not\in C$. Then under arbitrary deterministic strategy $\pi$ that satisfies $\E^{\pi}\conf{q}{l}<\infty$ we have $\E^{\pi}\tleave \geq \frac{l-V_{C}-1-\E^{\pi}\counter{\tleave}}{|\bar{x}_C|}$.%
\end{lemma}
\begin{proof}

Fix an arbitrary initial configuration $\conf{q}{l}$ and deterministic strategy $\pi$ with $\E^{\pi}\conf{q}{l}<\infty$.  %
Consider the following stochastic process $\{\maralt{i}\}_{i \geq 0}$:
\[
\maralt{i}
\coloneqq
\begin{cases}
\counter{i} +
\bar{\poten}_{\cstate{i}} - i\cdot \bar{x}_C
&
\text{if } \tleave\geq i \text{ and } \cstate{i}\in C,\\
\counter{i} + 1 + \bar{\poten}_{\cstate{i-1}} - i\cdot \bar{x}_C  
&
\text{if } \tleave\geq i \text{ and } \cstate{i}\not\in C,\\
\maralt{i-1} & \text{otherwise.}
\end{cases}
\]

We claim that $\{\maralt{i}\}_{i \geq 0}$ is a submartingale relative to the filtration $\{\mathcal{F}_i\}_{i\geq 0}$. The proof is again essentially the same as proof of Lemma 3 in \cite{BBEK:OC-games-termination-approx}. First, the value of $\maralt{i}(\omega)$ clearly depends only on finite prefix of $\omega$ of length $i$. Now let $u$ be any history of length $i$. If $\counter{j}(u)=0$ for some  $ 0\leq j< i$ or $\cstate{j} \not\in C$ for some  $ 0\leq j\leq i$ (i.e. if $\tleave(\omega)<i$ for all $\omega\in\runs{u}$), then clearly $\E^{\pi}(\maralt{i+1}\mid\runs{u})=\maralt{i}(u)$. 

Otherwise we denote $\conf{r}{j}$ the last configuration on $u$ and for every possible successor $\conf{r'}{j'}$ of $\conf{r}{j}$ in $\M^{\infty}_{\A}$ we set
\[
 p_{\conf{r'}{j'}} = \begin{cases} \pi(u)(\conf{r}{j}\btran{}\conf{r'}{j'}) & \text{if } r \in Q_1 \\
    \Prob(\conf{r}{j})(\conf{r}{j}\btran{}\conf{r'}{j'}) & \text{if } r \in Q_0.
\end{cases}
\]

Suppose that $r\in Q_1$ and that $\pi$ selects a transition to a configuration $\conf{r'}{j'}$ with $r' \not\in C$. Then 
\begin{align*}
\E^{\pi}(\maralt{i+1}\mid \runs{u}) &= \E^{\pi}(\counter{i+1} +1 + \bar{\poten}_{\cstate{i}} -(i+1) \cdot \bar{x}_{C} \mid \runs{u}) \\
&= \counter{i}(u) + \E^{\pi}(\underbrace{\counter{i+1} - \counter{i} - \bar{x}_{C}+1}_{\geq 0} + \bar{\poten}_{\cstate{i}}\mid \runs{u}) - i\cdot \bar{x}_C \\ 
&\geq \counter{i}(u) + \bar{\poten}_{\cstate{i}(u)} - i\cdot \bar{x}_C = \maralt{i}(u).
\end{align*}

On the other hand, if $r \in Q_0$ (in which case all successor configurations $\conf{r'}{j'}$ must satisfy $r'\in C$) or $r \in Q_1$ and $\pi$ selects transition that stays in $C$, then we have
\begin{align*}
 \E^{\pi}(\maralt{i+1}\mid \runs{u}) &= \E^{\pi}(\counter{i+1} + \bar{\poten}_{\cstate{i+1}} -(i+1) \cdot \bar{x}_{C} \mid \runs{u}) \\
 &= \counter{i}(u) + \E^{\pi}(\counter{i+1} - \counter{i} - \bar{x}_{C} + \bar{\poten}_{\cstate{i+1}}\mid \runs{u}) - i\cdot \bar{x}_C \\
 &= \counter{i}(u) \underbrace{- \bar{x}_C + \sum_{(r,k,r')\in\delta} p_{\conf{r'}{j'}}\cdot (k + \bar{\poten}_{r'})}_{\geq \bar{\poten}_{r} \text{ since } \left(\bar{x}_C,(\bar{\poten}_q)_{q\in C}\right) \text{ is a solution of $\mathcal{L}$} } - i\cdot \bar{x}_C \\
 &\geq \counter{i}(u) + \bar{\poten}_{\cstate{i}(u)} - i \cdot \bar{x}_C = \maralt{i}(u).
\end{align*}
Thus, $\{\maralt{i}\}_{i \geq 0}$ is indeed a submartingale. It is easy to see that $\{\maralt{i}\}_{i \geq 0}$ has bounded differences.

Clearly, the membership of every run $\omega$ in $\{\tleave \leq n\}$ depends only on finite prefix of $\omega$ of length $n$, and thus $\tleave$ is a stopping time relative to the filtration $\{\mathcal{F}_i\}_{i\geq 0}$. Also, for every run $\omega$ we have $\tleave(\omega) \leq T(\omega)$ and since we assume that $\E^{\pi}\conf{q}{i} < \infty$, we must also have $\E^{\pi}\tleave < \infty$. Thus the Optional stopping theorem applies and we have $\E^{\pi} \maralt{0}\leq \E^{\pi} \maralt{\tleave}$. But $\maralt{0} = l + \bar{\poten}_{q}$ and $\maralt{\tleave} \leq \E^{\pi}\counter{\tleave} + \max_{r \in C}\bar{\poten}_{r}+ 1 +|\bar{x}_C|\cdot \E^{\pi}\tleave$. This gives us $\E^{\pi}\tleave \geq (l + \bar{\poten}_{q} - \max_{r \in C}\bar{\poten}_{r} - 1 - \E^{\pi}\counter{\tleave})/|\bar{x}_{C}| \geq {(l-V_{C}-1-\E^{\pi}\counter{\tleave})}/{|\bar{x}_C|}$.
\qed

\end{proof} 

In the following we say that $q$ is a MEC state of $\A$ if it lies in some MEC of $\A$. Otherwise we say that $q$ is a non-MEC state.

We call state $q'$ a transient successor of state $q$ if both $q$ and $q'$ are non-MEC states and $q'$ is reachable from $q$ along a path that doesn't visit any MEC. We denote  $\trancard{\A}$ the maximal number of transient successors of any state in $\A$.

\begin{lemma}
Let $\A$ be arbitrary OC-MDP and let $q$ be arbitrary state of $\A$ that is not contained in any MEC. Then under arbitrary strategy $\pi$ the probability that, when starting in $q$, we will reach some MEC of $\A$ in at most $\trancard{\A}$ steps, is at least $p_{\min}^{\trancard{\A}}$.
\end{lemma}
\begin{proof}
We inductively define sets $H_0,H_1,\dots \subseteq 2^{|Q|}$. We set $H_0 = \{q\}$. Then, we construct $H_i$ from $H_{i-1}$ by initially setting $H_i = \emptyset$ and then performing the following operation \emph{for every set} $R \in H_{i-1}$: We find a state $q_R \in R$ such that $q_R$ is not contained in any MEC of $\A$ and $\{s \mid q_R \btran{} s\} \cap R = \emptyset$. If there is no such state in $R$, then we add $R$ to $H_i$. Otherwise:
 \begin{itemize}
 \item If $q_R$ is a stochastic state, then we set $R' = R \cup\{s \mid q_R \btran{} s  \}$ and add $R'$ to $H_i$. 
 \item If $q_R$ is a non-deterministic state, then we denote $\{s \mid q_R \btran{} s  \} = \{s_1,\dots,s_n\}$. After this, we create $n$ new sets $R_1 ,\dots R_n$, where $R_i = R \cup\{s_i\}$. Finally, we add sets $R_1,\dots,R_n$ to $H_i$.
 \end{itemize}
 For every $i$ and every $R \in H_i$ all the non-MEC states in $R$ are transient successors of $q$. Thus, $H_{\trancard{\A}}=H_{\trancard{\A}+1}$. We claim that every set $R\in H_{\trancard{\A}} $ must contain at least one MEC-state of $\A$. Assume, for the sake of contradiction, that there is some $R\in H_{\trancard{\A}}$ containing only non-MEC-states. Then $R$ satisfies the following: for every state $q \in R$, if $q$ is non-deterministic then there is at least one state $s \in R$ such that $q \btran{} s$; otherwise, if $q$ is stochastic, then $\{s \mid q \btran{} s\} \subseteq R$. This also means, that restriction of $\A$ to set $R$, i.e. the tuple $\A_R=(R,(R\cap Q_0,R\cap Q_1),\delta\cap (R\times \{+1,0,-1\}\times R),\{P_q\}_{q\in R\cap Q_0})$, is again a OC-MDP. As every OC-MDP, the $\A_R$ also contains at least one MEC $E$, which must be contained in some MEC of $\A$. This contradicts the assumption that $R$ contains only non-MEC states.%

 Now let $\pi$ be arbitrary strategy and $i \geq 0$. Denote $R_i(\pi)$ the set of states that are, when starting in $q$, reached under $\pi$ in at most $i$ steps. From the construction of $H_{i}$ it follows by straightforward induction, that there is some set $R \in H_{i}$ such that $R \subseteq R_i(\pi)$. In particular, there is some set $R \in H_{\trancard{\A}}$ such that $R \subseteq R_{\trancard{\A}}(\pi)$. Since $R$ must contain at least one MEC-state of $\A$, there is some history $u$ of length at most $\trancard{\A}$ such that $u$ reaches a MEC state and $\Pr{\pi}{q}{\runs{u}}>0$. Then clearly $\Pr{\pi}{q}{\runs{u}}\geq p_{\min}^{\trancard{\A}}$ and this proves the lemma. 
 \qed
\end{proof}

\begin{corollary}
\label{cor:mdp-transient-time}
 Let $q$ be an arbitrary state of $\A$. Denote $T_M$ the random variable on runs starting in $q$ that returns the first point in time, when some MEC of $\A$ is reached. Then for arbitrary strategy $\pi$ and every $k \geq 1$ we have $\Pr{\pi}{q}{T_M\geq k} \leq 4d^k$, where $d = \exp(-p_{\min}^{\trancard{\A}}/\trancard{\A} )$.
\end{corollary}
\begin{proof}
If $p_{\min}=1$ then $\Pr{\pi}{q}{T_M > \trancard{\A}}=0$ and thus the Lemma trivially holds. Otherwise we have $p_{\min} \leq 1/2$.
From the previous lemma we immediately see that \mbox{$\Pr{\pi}{q}{T_M \geq k} \leq (1-p_{\min}^{\trancard{\A}})^{\lfloor\frac{k-1}{\trancard{\A}} \rfloor}$}.
We can now compute
\begin{align*}
 \Pr{\pi}{q}{T_M \geq k} &\leq (1-p_{\min}^{\trancard{\A}})^{\lfloor\frac{k-1}{\trancard{\A}} \rfloor}
 \leq (1-p_{\min}^{\trancard{\A}})^{\frac{k}{\trancard{\A}} -2 } = \frac{(1-p_{\min}^{\trancard{\A}})^{\frac{k}{\trancard{\A}}}}{(1-p_{\min}^{\trancard{\A}})^2} \leq 4(1-p_{\min}^{\trancard{\A}})^{\frac{k}{\trancard{\A}}} \leq 4d^k.
\end{align*}
\qed

\end{proof}

Let $r$ and $r'$ be two states of $\A$ that lie in the same MEC $C$. Then clearly $t_r = t_{r'}$. We will denote $t_C$ the common value $t_r$ of all $r$ in $C$.

We now prove Proposition \ref{prop:upper-general-slowbound} for MEC-acyclic OC-MDPs. We say that a OC-MDP $\A$ is MEC-acyclic if there is no cycle in $\A$ containing states from two different MECs. Equivalently, one can say that $\A$ is MEC-acyclic if no run in $\A$ returns to some MEC once it leaves this MEC. The \emph{height} of a state $q$ in MEC-acyclic OC-MDP $\A$, which we denote $\height{q}$, is the maximal number of MECs visited by any path starting in $q$. The height of a given MEC $C$ is the common height of all its states.

For \emph{any} OC-MDP $\A$ we denote $\size{\A_{\max}} = \max\{\size{\A_C}\mid C \in \MEC(\A)\}$. 

\begin{lemma}
\label{lem:mec-tree-lemma}
 Let $\A$ be a MEC-acyclic OC-MDP. Then there is a number $K = \exp\left(\size{\A_{\max}}^{\calO(1)} \right)\cdot \calO(\trancard{\A}/p_{\min}^{\trancard{\A}})$ such that the following holds for every memoryless deterministic strategy $\pi$ and every initial configuration $\conf{q}{i}$:
\begin{equation}
\label{eq:mec-ac-lower}
 \E^{\pi}T \geq \frac{i}{|t_q|} - K\cdot \E^{\pi}\switch.
\end{equation}
Moreover, $K$ is computable in time polynomial in  $\size{\A_{\max}} \cdot \log(p_{\min})\cdot \trancard{\A}$ by algorithm that takes as an input number $\trancard{\A}$ and set of strongly connected OC-MDPs \mbox{$\{\A_C\mid C \in\MEC(\A)\}$.}%
\end{lemma}
\begin{proof}
Recall that we denote $d = \exp(-p_{\min}^{\trancard{\A}}/\trancard{\A} )$ and set $$K=\max\left\{\frac{4}{(1-d)^2\cdot|\bar{x}_0|},\frac{1+\max_{C \in \MEC(\A)}V_{C}}{|\bar{x}_0|}  \right\}.$$

The asymptotic upper bound on $K$ is easy to check, since numbers $\bar{x}_0$ and $V_C$ for $C \in \MEC(\A)$ are computed by solving linear program $\mathcal{L}$ for MECs of $\A$; also recall that $1/(1-d)^2 \leq 4\trancard{\A}^2/{p_{\min}^{2\trancard{\A}}}$ by standard calculus computation. This also shows that $K$ can be computed in time polynomial in $\size{\A_{\max}} \cdot \log(p_{\min})\cdot \trancard{\A}$ if we know numbers $\trancard{\A}$ and $p_{\min}$ and OC-MDPs $\A_C$ for every $\MEC$ $\C$ of $\A$.

Note that in \emph{every} OC-MDP we have $\E^{\pi}\switch<\infty$ under any strategy $\pi$ (by Lemma \ref{lem:switches-bound}).
Therefore, both inequalities trivially hold if $\E^{\pi}q(i) = \infty$. From now on we will assume that $\E^{\pi}q(i) < \infty$. In particular, we assume that under $\pi$ the configuration with zero counter is reached almost surely from $q(i)$.
We proceed by induction on $\height{q}$. For every height we will prove the inequality separately for $q$ being a non-MEC state  and  MEC-state, respectively.

To start the induction, suppose that $q$ lies in MEC $C$ of height $1$. But then there are no transitions leaving $C$. In particular, we have $\E^{\pi}\switch=1$. From part (B1) of Proposition \ref{prop:charact} and from $K \geq \frac{V_{C}}{|\bar{x}_0|}$ we have 
\begin{align*}
 \E^{\pi} q(i) &\geq \frac{i - V_{C}}{|\bar{x}_C|} \geq  \frac{i}{|t_q|} - K.
\end{align*}
The second equality holds because for state $q$ that lies in MEC $C$ with no outgoing transitions we have $t_q = \bar{x}_C$. 

Suppose now that $q$ is a non-MEC-state of height $h$ and that \eqref{eq:mec-ac-lower} holds for all MEC-states of height at most $h$.
 
Denote $F_C$ the event that the first MEC encountered on a run is $C$. Note that all MECs with $\Pr{\pi}{q(i)}{F_C}>0$ have height at most $h$. Denote $D$ the union of all MECs $C$ with $\Pr{\pi}{q(i)}{F_C}>0$. Similarly to previous proofs we can write $\E^{\pi}q(i) = E^{\pi} (T_1 + T_2)$ where $T_1$ returns the first point in time when the run hits either $D$ or a configuration with a zero counter and $T_2$ returns time to hit a configuration with a zero counter after hitting $D$ (or 0, if the run terminates before hitting $D$ or never hits $D$ at all). Since both these random variables are non-negative, it suffices to prove the required bound \eqref{eq:mec-ac-lower} for $\E^{\pi} T_2$. %

As in previous proofs, we use the notation $D_{m}$ (for $m>0$) for the set of all runs that do not terminate before reaching $D$ and at the same time they reach $D$ with counter value $m$. (Also recall that we denote $D_0$ set of runs that terminate before or in the exact moment of reaching $D$.) %
Moreover we denote $D_{m}^C$ the event $F_C \cap D_{m}$. Finally, we denote 
 \[
 \lowerrem{l}{j}{C}:=\frac{j}{|t_C|} - K \cdot(\E^{\pi}(\switch\mid D_{l}^C)-1).
  \]

Clearly $\sum_{\substack{C \in \MEC(\A), l\geq 0}} \Pr{\pi}{\conf{q}{i}}{D^C_l} = \sum_{C \in \MEC(\A)}\Pr{\pi}{\conf{q}{i}}{F_C}=1$.

We have
\begin{equation}
 \label{eq:gen-lb1}
\E^{\pi} T_2 = \sum_{C \in \MEC(\A)} \Pr{\pi}{q(i)}{F_C}\cdot \E^{\pi}(T_2\mid F_C).
\end{equation}

We can write 
\begin{align}
 \Pr{\pi}{q(i)}{F_C}\cdot\E^{\pi}(T_2\mid F_C) &= \sum_{l=0}^{\infty}\E^{\pi}(T_2\mid D_{l}^C)\cdot\Pr{\pi}{q(i)}{D_{l}^C}\label{eq:gen-lb2}.
\end{align}

By induction hypothesis we have for every $l \geq 0$%
\begin{align}
 \E^{\pi}(T_2\mid D^C_{l}) &\geq \frac{l}{|t_C|} - K\cdot( \E^{\pi}(\switch\mid D_l^C) -1 ) = \lowerrem{l}{l}{C}.\label{eq:gen-lb3}
\end{align}

Especially for every $l\geq i$ we have
\begin{equation}
 \label{eq:gen-lb5}
 \E^{\pi}(T_2\mid D^C_{l}) \geq \frac{i}{|t_C|} - K\cdot( \E^{\pi}(\switch\mid D_l^C) -1 )= \lowerrem{l}{i}{C}.
\end{equation}

Further, if we denote $g_l = i-l$ then for $l<i$ we can write
\begin{align}
  \E^{\pi}(T_2\mid D^C_{l}) &\geq \frac{(i-g_l)}{|t_C|} - K\cdot( \E^{\pi}(\switch\mid D_l^C) -1 ) = \lowerrem{l}{i}{C} - \frac{g_l}{|t_C|}. \label{eq:gen-lb6}
\end{align}

We can now plug \eqref{eq:gen-lb5} and \eqref{eq:gen-lb6} into \eqref{eq:gen-lb2} and compute
\begin{align}
 \E^{\pi} T_2 &\geq \sum_{C \in \MEC(\A)}\left(\sum_{l=i}^{\infty}\left(\lowerrem{l}{i}{C}\cdot\Pr{\pi}{q(i)}{D_{l}^C}\right) + \sum_{l=0}^{i-1} \left(\left(\lowerrem{l}{i}{C} - \frac{g_l}{|t_C|}\right)\cdot \Pr{\pi}{q(i)}{D_{l}^C} \right)\right) \nonumber \\ 
  &=\sum_{C \in \MEC(\A)}\left(\sum_{l=0}^{\infty}\left(\lowerrem{l}{i}{C}\cdot\Pr{\pi}{q(i)}{D^C_l} \right)- \sum_{l=0}^{i-1}\frac{g_l}{|t_C|}\cdot\Pr{\pi}{q(i)}{D_{ l}^C} \right) \nonumber \\
  &=i \cdot \underbrace{\sum_{C \in \MEC(\A)}\left(\frac{\Pr{\pi}{q(i)}{F_C}}{|t_C|}\right)}_{\geq \frac{1}{|t_q|}} - K\cdot \sum_{\substack{C \in \MEC(\A),\\ l \geq 0}}\left(  \E^{\pi}(\switch\mid D_l^C) -1 \right)\cdot \Pr{\pi}{q(i)}{D_l^C}\nonumber \\ &\text{\ \ \ }-
  \sum_{C \in\MEC(\A)}\sum_{l=0}^{i-1}\left(\frac{g_l}{|t_C|}\cdot\Pr{\pi}{q(i)}{D_{ l}^C} \right)\nonumber \\
  &\geq \frac{i}{|t_q|} - K\cdot \E^{\pi}\switch + K  -\sum_{C\in \MEC{(\A)}}\sum_{l=0}^{i-1}\left(\frac{g_l}{|\bar{x}_0|}\cdot\Pr{\pi}{q(i)}{D_{ l}^C}\right)\nonumber \\
  &= \frac{i}{|t_q|} - K\cdot \E^{\pi}\switch + K  -\frac{\sum_{l=0}^{i-1}\left(g_l \cdot\Pr{\pi}{q(i)}{D_{ l}}\right)}{|\bar{x}_0|}. \label{eq:gen-lb4}
\end{align}

From Corollary \ref{cor:mdp-transient-time} we have 
\begin{equation}
\label{eq:gen-lb7}
 \frac{\sum_{l=0}^{i-1}\left(g_l\cdot\Pr{\pi}{q(i)}{D_{ l}}\right)}{|\bar{x}_0|}\leq \frac{4\cdot\sum_{l=0}^{i-1}g_l d^{g_l}}{|\bar{x}_0|}  = \frac{ 4\cdot\sum_{g_l=1}^{i}g_l d^{g_l}}{|\bar{x}_0|} \leq \frac{\frac{4}{(1-d)^2}}{|\bar{x}_0|} \leq {K},
\end{equation}
since no run in $D_l$ , for $l<i$, can hit $D$ in less than $g_l$ steps.

This gives us $K - \sum_{l=0}^{i-1}\frac{g_l}{|\bar{x}_0|}\cdot \Pr{\pi}{q(i)}{D_{ l}}\geq 0$ and together with \eqref{eq:gen-lb4} we have
\begin{align*}
 \E^{\pi} T_2 &\geq \frac{i}{|t_q|} - K \cdot \E^{\pi}\switch, %
\end{align*}
which proves that \eqref{eq:mec-ac-lower} holds for $q$.

Suppose now that $q$ lies in MEC $C$ of height $h$ and that \eqref{eq:mec-ac-lower} holds for all states of height $h-1$. The inequality \eqref{eq:mec-ac-lower} especially holds for all states $q' \in Q_{fin}\setminus C$ such that there is a transition from $p$ to $q'$ for some $p \in C$. We will call every such state $q'$ a \emph{$C$-gate} and denote $G(C)$ the set of all $C$-gates. From the definition of $t_q$ it follows that $\frac{1}{|t_q|} \leq \frac{1}{|\bar{x}_C|}$ and $\frac{1}{|t_q|} \leq \frac{1}{|t_{q'}|}$ for any $C$-gate $q'$. 

We can again express $T$ as a sum of $T_1$ and $T_2$, where $T_1$ returns the first point in time when the run visits configuration $r(l)$ with either $r\not\in C$ or $l=0$, and $T_2$ returns time to visit a configuration with a zero counter after leaving $C$ (or 0, if the run terminates before leaving $C$ or never leaves $C$ -- formally we again have $T_2(\omega)=-T_1(\omega) + T(\omega)$ if $T_1(\omega)<\infty$ and $T_2(\omega)=0$ otherwise). 
From Lemma \ref{lem:general-mec-azuma} we have
\begin{equation}
\label{eq:gen-lb8}
\E^{\pi}(T_1) \geq \frac{i-V_{C}-1-\E^{\pi}\counter{T_1}}{|\bar{x}_C|} \geq \frac{i-\E^{\pi}\counter{T_1}}{|t_q|}-\frac{V_{C}+1}{|\bar{x}_0|}.
\end{equation}

Now consider $T_2$. For state $q'$ not contained in $C$ we denote $F^{q'}_l$ the set of all runs $\omega$ that visit configuration $q'(l)$ when they leave $C$ for the first time, i.e. $\omega \in F^{q'}_l$ iff $\cstate{T_1}(\omega)=q'$ and $\counter{T_1}(\omega)=l$. Note that for every $q'$ such that $\Pr{\pi}{\conf{q}{i}}{F^{q'}_l}>0$ we must have $q' \in G(C)$. If we denote $F_l = \bigcup_{q' \in G(C)} F^{q'}_l$, then it is easy to see that $\E^{\pi} \counter{T_1} = \sum_{l \in \Nset}l\cdot \Pr{\pi}{\conf{q}{i}}{F_l}$. Finally, denote $\leave_C$ the event that the run leaves $C$ at least once (i.e. $\omega \in \leave_C$ iff $\omega \in D_l^{q'}$ for some $l$ and $q'$).
We have
\begin{align}
\label{eq:gen-lb9}
 \E^{\pi}T_2 &= \sum_{\substack{q' \in G(C),\\l \geq 0}} \E^{\pi}(T_2\mid F^{q'}_l)\cdot\Pr{\pi}{\conf{q}{i}}{F^{q'}_l} \nonumber \\ 
  &\geq \sum_{\substack{q' \in G(C),\\l \geq 0}} \left(\left(\frac{l}{|t_{q'}|}-K\cdot(\E^{\pi}(\switch \mid F^{q'}_l)-1)\right)\cdot\Pr{\pi}{\conf{q}{i}}{F^{q'}_l}\right)\nonumber \\
  &\geq \sum_{\substack{q' \in G(C),\\l \geq 0}} \left(\frac{l}{|t_{q}|}-K\cdot(\E^{\pi}(\switch \mid F^{q'}_l)-1)\right)\cdot\Pr{\pi}{\conf{q}{i}}{F^{q'}_l}\nonumber\\ 
 &= \frac{\sum_{l\geq 0}\left( l\cdot\Pr{\pi}{\conf{q}{i}}{F_l} \right) }{|t_q|} -K\cdot\left(\E^{\pi}(\switch\mid \leave_C)\cdot \Pr{\pi}{\conf{q}{i}}{\leave_C} -\Pr{\pi}{\conf{q}{i}}{\leave_C}\right) \nonumber \\
 &= \frac{\E^{\pi}\counter{T_1}}{|t_q|}-K\cdot\left(\E^{\pi}(\switch\mid \leave_C)\cdot \Pr{\pi}{\conf{q}{i}}{\leave_C} -\Pr{\pi}{\conf{q}{i}}{\leave_C}\right),
\end{align}
where the inequality on the second line follows from induction hypothesis.

Denote $\overline{\leave_C}$ the complement of $\leave_C$. We trivially have $\E^{\pi}(\switch\mid \overline{\leave_C})\geq 1$.
Putting \eqref{eq:gen-lb8} and \eqref{eq:gen-lb9} together we obtain
\begin{align*}
 \E^{\pi}\conf{q}{i} &\geq \frac{i}{|t_q|} -K\cdot \E^{\pi}(\switch\mid \leave_C)\cdot\Pr{\pi}{\conf{q}{i}}{\leave_C}  + K\cdot \Pr{\pi}{\conf{q}{i}}{\leave_C} - \underbrace{\frac{V_{C}+1}{|\bar{x}_0|}}_{\leq K} \\
  &\geq \frac{i}{|t_q|} -K\cdot \E^{\pi}(\switch\mid \leave_C)\cdot\Pr{\pi}{\conf{q}{i}}{\leave_C} -K \cdot {1} \cdot (1-\Pr{\pi}{\conf{q}{i}}{\leave_C}) \\
  &\geq \frac{i}{|t_q|} -K\cdot \left(\E^{\pi}(\switch\mid \leave_C)\cdot\Pr{\pi}{\conf{q}{i}}{\leave_C} + \E^{\pi}(\switch\mid \overline{\leave_C})\cdot\Pr{\pi}{\conf{q}{i}}{\overline{\leave_C}} \right) \\
  &= \frac{i}{|t_q|} - K \cdot \E^{\pi}\switch.
\end{align*}
Thus, \eqref{eq:mec-ac-lower} indeed holds for $q$. \qed
\end{proof}

We will now finish the proof of Proposition \ref{prop:upper-general-slowbound} for arbitrary OC-MDP with $Q=Q_{fin}$.

To achieve this, for arbitrary OC-MDP $\A$ and any natural number $k$ we define a new MEC-acyclic OC-MDP $\A(k)$ of height $k+1$; we will augment states of $\A$ with additional information, that will allow us to remember number of visits of MECs. Once we know that we have left a MEC for the $k$-th time, we allow to switch to a new state with a counter-decreasing self-loop. 
To be more specific, call the transition $q \btran{} q'$ a \emph{crossing}, if there exists a MEC $C$ such that $q \in C$, $q'\not\in C$. %
Then for $\A=(Q,(Q_0,Q_1),\delta,P)$ we set $\A(k) = (Q^k,(Q^k_0,Q^k_1),\delta^k,P^k)$, where $Q^k = \{(q,l)\mid q \in Q, 1 \leq l \leq k \} \cup \{\bot\}$, and 
\begin{align*}
\delta^k&=\{((q,l),i,(q',l)) \mid (q,i,q')\in \delta \text{ and $(q,i,q')$ is not a crossing }\}\\  &\cup \{((q,l),i,(q',l-1))\mid l>1,~(q,i,q') \in \delta \text{ is a crossing } \}\\&\cup \{((q,1),i,\bot)\mid \exists q' \text{ such that }  ~(q,i,q')\in \delta \text{ is a crossing } \}\\ &\cup \{(\bot,-1,\bot)\}.
\end{align*}
Partition of states $(Q^k_0,Q^k_1)$ and probability distribution $P^k$ is derived from $\A$ in obvious way, we just specifically put $\bot \in Q_0^k$. 

Slightly abusing the notation we denote ${t}_{q^k}$ the minimal trend achievable from state $(q,k)$ in $\A(k)$.

For every deterministic strategy $\pi$ in $\A$ there is naturally corresponding deterministic strategy $\pi(k)$ in ${\A}(k)$, formally defined as follows: for any history $\bar{H}=(q_0,l_0)(j_0)\dots(q_m,l_m)(j_m)$ in ${\A}(k)$ we denote $q'(j')$ configuration of $\A$ such that $\pi(q_0(j_0)\dots q_m(j_m))$ selects transition leading to configuration $q'(j')$; then we define $(\pi(k))(H)$ to select transition leading to configuration $c$ of ${\A}(k)$ such that
\[
 c = \begin{cases}
					   \bot(j') & \text{if } l_m = 1 \text{ and } q_m \btran{} q' \text{ is a crossing,} \\
                                           (q',l_m - 1)(j') & \text{if } l_m > 1 \text{ and } q_m\btran{} q' \text{ is a crossing,} \\
					   (q',l_m)(j') & \text{otherwise. }
                                           \end{cases}
\]

To differentiate between computations in $\A$ and ${\A}(k)$, 
we again slightly abuse notation and denote $\mathbb{P}^{\pi(k)}$ and ${\E}^{\pi(k)}$ the probability and expected value, respectively, computed in $\A(k)$ under strategy $\pi(k)$. Note that if $\pi$ is memoryless deterministic, then $\pi(k)$ is also memoryless deterministic.

It is clear that for any strategy $\pi$ in $\A$ and any $k \geq 1$ we have $\E^{\pi}\conf{q}{i} \geq {\E}^{{\pi}(k)} \conf{(q,k)}{i}$. We can thus use the Lemma \ref{lem:mec-tree-lemma} to show that for any memoryless deterministic strategy $\pi$ and any $k \geq 1$ we have 
\begin{equation}\label{eq:near-lower-bound}\E^{\pi} q(i) \geq \frac{i}{|{t}_{q^k}|}-K\cdot \E^{\pi(k)}_{(q,k)(i)}\switch,\end{equation} 
for a suitable number $K$. Note that for every $k$ the MECs of $\A(k)$ are exactly copies of MECs of $\A$ (with the exception of MEC $\{\bot\}$). It is also easy to see that $\trancard{\A(k)} \leq |Q|$, for every $k$, and that $p_{\min}$ is the same in $\A$ and $\A(k)$ for every $k$. By Lemma \ref{lem:mec-tree-lemma} this means that $K\in \exp\left(\size{\A} ^{\calO(1)}\right) $ can be chosen the same for every $k$ and that it can be computed by a polynomial-time algorithm that takes $\A$ as its input. (This is important observation: we do not have to construct any MEC-acyclic OC-MDP in order to compute $K$.)%

To finish the proof of Proposition \ref{prop:upper-general-slowbound} it suffices to show that $$\lim_{k \rightarrow \infty}\left( \frac{i}{|{t}_{q^k}|}-K\cdot \E^{\pi(k)}_{(q,k)(i)}\switch \right)\geq \frac{i}{|t_q|} - K\cdot \E^{\pi}_{q(i)}W.$$ This is done in following two lemmas.

\begin{lemma}
 We have $\lim_{k \rightarrow \infty} \frac{1}{|t_{q^k}|} = \frac{1}{|t_q|}$.
\end{lemma}
\begin{proof}
 For any $k \geq 1$ we clearly have $|t_q|^{-1} \geq |t_{q^k}|^{-1}$, so it suffices to prove that $\lim_{k \rightarrow \infty} \frac{1}{|t_{q^k}|} \geq \frac{1}{|t_q|}$.
 Fix arbitrary $k\geq 1$.

Consider the ``fast`` counterless strategy $\rho^k$ from Proposition \ref{prop:lower-main}, that realizes the minimal trend $t_{q^k}$ in $\A(k)$. We define a new strategy $\rho'$ in $\A$ as follows: Initially, $\rho'$ behaves exactly as $\rho^k$, simply omitting the information on current depth stored in states of ${\A}(k)$. When strategy $\rho^k$ prescribes to switch to state $\bot$, the strategy $\rho'$ starts to behave as the ``fast`` counterless strategy $\sigma$ in $\A$ from Proposition \ref{prop:lower-main}. 

Denote $\reach_k(\bot)$ the event that run in $\A(k)$ reaches state $\bot$.
Simple computation, which uses the fact that, apart from $\{\bot\}$, all MECs of $\A(k)$ are copies of MECs of $\A$, reveals that 
\[
 \underbrace{\sum_{C \in \MEC(\A)}\frac{\Pr{\rho'}{q(i)}{M_C}}{|\bar{x}_C|}}_{\geq \frac{1}{|t_q|}} - \underbrace{\sum_{{C} \in \MEC(\A(k))}\frac{\Pr{\rho^k}{(q,k)(i)}{M_{C}}}{|\bar{x}_{{C}}|}}_{=\frac{1}{|t_{q^k}|}} \leq \Pr{\rho^k}{(q,k)(i)}{\reach_k(\bot)}\cdot\left(\frac{1}{|\bar{x}_0|} -1\right).
\]

From the construction of $\A(k)$ it easily follows that $\Pr{\rho^k}{(q,k)}{\reach_k(\bot)} \leq \Pr{\rho'}{q}{W\geq k}$. By Lemma \ref{lem:switches-bound} we have that $\Pr{\rho'}{q}{W\geq k} \rightarrow 0$ as $k \rightarrow \infty$. This gives us 
\begin{align*}
 \frac{1}{|t_q|} - \lim_{k \rightarrow \infty} \frac{1}{|t_{q^k}|} \leq 0,
\end{align*}
which proves the lemma. \qed

\begin{lemma}
 We have $\lim_{k \rightarrow \infty} \E^{\pi(k)}_{(q,k)(i)}\switch = \E^{\pi}_{q(i)}\switch$.
\end{lemma}
\begin{proof}
Fix arbitrary $k \geq 1$.
We have $\E^{\pi(k)}_{(q,k)(i)}\switch = \sum_{l \geq 1} l\cdot \Pr{\pi(k)}{\conf{(q,k)}{i}}{W=l}$ and \mbox{$\E^{\pi}_{q(i)}\switch=\sum_{l \geq 1} l\cdot \Pr{\pi}{\conf{q}{i}}{W=l}$}. From the construction of $\A(k)$ it easily follows that for all $l \leq k$ we have $\Pr{\pi(k)}{\conf{(q,k)}{i}}{W=l}=\Pr{\pi}{\conf{q}{i}}{W=l}$ and thus
\begin{align*}
 |\E^{\pi}_{q(i)}\switch - \E^{\pi(k)}_{(q,k)(i)}\switch| &\leq \sum_{l=k}^{\infty}l\cdot |\Pr{\pi}{\conf{q}{i}}{W=l} - \Pr{\pi(k)}{\conf{(q,k)}{i}}{W=l}|  \\ 
  &\leq \sum_{l=k}^{\infty}l\cdot \Pr{\pi}{\conf{q}{i}}{W=l} + \sum_{l=k}^{\infty}l\cdot \Pr{\pi(k)}{\conf{(q,k)}{i}}{W=l}.
\end{align*}
From Lemma \ref{lem:switches-bound} we have that $\Pr{\pi}{\conf{q}{i}}{W=l}\leq b\cdot c^l$ for suitable numbers $b$ and \mbox{$0<c<1$}. Moreover, 
$\Pr{\pi(k)}{\conf{(q,k)}{i}}{W=l}=0$ for all $l\geq 2\cdot(k+1)$. Also, since $\Pr{\pi(k)}{\conf{(q,k)}{i}}{W=l}=\Pr{\pi}{\conf{q}{i}}{W=l}$ for $l \leq k$, we have $\sum_{l=k}^{\infty} \Pr{\pi(k)}{\conf{(q,k)}{i}}{W=l} = \sum_{l=k}^{\infty} \Pr{\pi}{\conf{q}{i}}{W=l}\leq b\cdot \sum_{l=k}^{\infty}c^l$. Thus we can write $|\E^{\pi}_{q(i)}\switch - \E^{\pi(k)}_{(q,k)(i)}\switch|\leq b\cdot \sum_{l=k}^{\infty}  c^l + 2\cdot(k+1)\cdot b \cdot \sum_{l=k}^{\infty}c^l\leq 3\cdot(k+1)\cdot b \cdot \sum_{l=k}^{\infty}c^l $. From standard results on power series we know that $$\lim_{k \rightarrow \infty}(k+1) \cdot \sum_{l=k}^{\infty} c^l \leq \lim_{k \rightarrow \infty}\sum_{l=k}^{\infty}(l+1)\cdot c^l= 0$$ and thus also $\lim_{k \rightarrow \infty}|\E^{\pi}_{q(i)}\switch - \E^{\pi(k)}_{(q,k)(i)}\switch|=0$. This proves the lemma. \qed
\end{proof}

\end{proof}

\subsection{Proofs of Section~\ref{sec-lower}}
\label{app-lower}

\tikzstyle{max}=[thick,draw,minimum size=1.4em,inner sep=0em]
\tikzstyle{ran}=[circle,thick,draw,minimum size=1.4em,%
    inner sep=0em]
\tikzstyle{ent}=[circle,thick,fill=lightgray,draw,minimum size=1.4em,%
    inner sep=0em]
\tikzstyle{tran}=[thick,draw,->,>=stealth]
\tikzstyle{ac}=[circle,thick,draw,color=mybrown,minimum size=1ex,%
    inner sep=0em,fill]
\tikzstyle{act}=[circle,thick,draw,color=mybrown,minimum size=1ex,%
    inner sep=0em,fill]

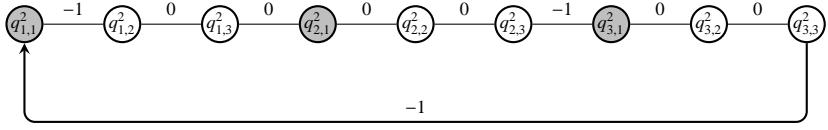
\begin{figure}[t]
\centering
\begin{tikzpicture}[x=1.3cm,y=1.3cm,font=\scriptsize]
\foreach \i/\s/\j/\k in {1/ent/1/1,2/ran/1/2,3/ran/1/3,4/ent/2/1,5/ran/2/2,%
   6/ran/2/3,7/ent/3/1,8/ran/3/2,9/ran/3/3}{%
     \node(q\i) at (\i,0) [\s] {$q^2_{\j,\k}$}; 
}
\foreach \i/\j/\l in {1/2/-1, 2/3/0, 3/4/0, 4/5/0, 5/6/0, 6/7/-1, 7/8/0,%
   8/9/0}{%
     \draw (q\i)  to [tran]  node[above] {$\l$}  (q\j);  
}
\draw (q9) [tran, rounded corners] -- +(0,-1) -- node[above] {$-1$} +(-8,-1)  
   -- (q1);
\end{tikzpicture}
\caption{The gadget for $x_2$ when $n=2$. Shadow states are the entry 
  points.}
\label{fig:hard}
\end{figure}

\begin{lemma}
\label{lem-lower}
  Given a propositional formula $\varphi$ in CNF, one can compute 
  a OC-MDP $\A$, a configuration $\conf{p}{K}$ of $\A$, and a number
  $N$ in time polynomial in $\size\varphi$ such that
  \begin{itemize}
  \item $N \leq |Q| \cdot K$, where $Q$ is the set of control 
     states of $\A$;
  \item if $\varphi$ is satisfiable, then $\val(\conf{p}{K}) = N-1$;
  \item if $\varphi$ is not satisfiable, then $\val(\conf{p}{K}) = N$.
  \end{itemize}
\end{lemma}
\begin{proof}
  Let $\varphi \equiv C_1 \wedge \cdots \wedge C_n$ where
  $C_1,\ldots,C_n$ are clauses over propositional variables
  $x_1,\ldots,x_m$. We may safely assume that $n \geq 5$. 
  Let $\pi_1,\ldots,\pi_m$ be the first $m$
  prime numbers. For every $x_i$, where $1 \leq i \leq m$, we 
  construct the gadget shown in Fig.~\ref{fig:hard}. That is, we
  fix $\pi_i \cdot (n+1)$ fresh stochastic control states 
  $q^i_{j,\ell}$, where $1 \leq j \leq \pi_i$ and $1 \leq \ell \leq n+1$,
  and connect them by transitions in the following way:
  \begin{itemize}
  \item $q^i_{1,1} \ltran{-1} q^i_{1,2}$,\quad 
     $q^i_{1,\ell} \ltran{0} q^i_{1,\ell+1}$ for all $2 \leq \ell \leq n$, 
     \quad $q^i_{1,n+1} \ltran{0} q^i_{2,1}$;\\[-1.5ex]
  \item for all $2 \leq j \leq \pi_i$ we include the following transitions:
     \begin{itemize}
        \item $q^i_{j,\ell} \ltran{0} q^i_{j,\ell+1}$ for all $1 \leq \ell \leq n$,
        \item $q^i_{j,n+1} \ltran{-1} q^i_{j',1}$, where $j'$ is either
           $j+1$ or $1$ depending on whether $j < \pi_i$ or not, respectively. 
     \end{itemize}
  \end{itemize}
  Since each $q^i_{j,\ell}$ has exactly one successor, all of the above
  transitions have probability one. Also note that the total size of
  the constructed gadgets is polynomial in $\size\varphi$ 
  because $\sum_{i=1}^m \pi_i$
  is $\calO(m^2\, \log m)$ (see, e.g., \cite{BS:book-number}).

  The control states of the form $q^i_{j,1}$, where $1 \leq j \leq \pi_i$,
  are called the \emph{entry points} for $x_i$. Note that in $q^i_{j,1}$,
  the counter is decremented in just one transition, while in the
  other entry points we need $n+1$ transitions to decrement the counter.

  An important technical observation about the entry points is the 
  following: For every $k \geq 1$ and $1 \leq i \leq n$, there is
  exactly one \emph{optimal} entry point $q^i_{j,1}$ such that
  $\val(\conf{q^i_{j,1}}{k}) = k(n+1) - n$, and for the other entry
  points $q^i_{j',1}$ we have that $\val(\conf{q^i_{j',1}}{k}) = k(n+1)$.
  To see this, consider the (unique) $k'$ such that $1 \leq k' \leq \pi_i$
  and $k = k' + c \cdot \pi_i$ for some $c \geq 0$. We put
  $j = 1$ if $k' =1$, otherwise $j = \pi_i - k' +2$. Now one can 
  easily verify (with the help of Fig.~\ref{fig:hard}) that 
  $\val(\conf{q^i_{j,1}}{k}) = k(n+1) - n$, and 
  $\val(\conf{q^i_{j',1}}{k}) = k(n+1)$ for the other entry points 
  $q^i_{j',1}$. 

  Every $k \geq 1$ encodes a unique assignment 
  $\nu_k : \{x_1,\ldots,x_m\} \rightarrow \{\mathit{true},\mathit{false}\}$
  defined as follows: For every $1\leq i \leq m$ we put 
  $\nu_k(x_i) = \mathit{true}$ iff $q^i_{1,1}$ is the optimal entry 
  point for~$k$. Also observe that for every assignment
  $\nu : \{x_1,\ldots,x_m\} \rightarrow \{\mathit{true},\mathit{false}\}$ 
  there is some $k \leq \prod_{i = 1}^{m} \pi_i$ such that $\nu = \nu_k$.

  We proceed by encoding the structure of $C_1,\ldots,C_n$. For each clause
  \mbox{$C_\ell \equiv y_{i_1} \vee \cdots \vee y_{i_t}$}, where every $y_{i_h}$
  is either $x_{i_h}$ or $\neg x_{i_h}$, we fix a fresh non-deterministic
  control state $c_\ell$ and add the following transitions for every 
  $1 \leq h \leq t$:
  \begin{itemize}
  \item if $y_{i_h} \equiv x_{i_h}$, then we add a transition 
     $c_\ell \ltran{0} q^{i_h}_{1,1}$;  
  \item if $y_{i_h} \equiv \neg x_{i_h}$, then we add a transition 
     $c_\ell \ltran{0} q^{i_h}_{j,1}$ for every $2 \leq j \leq \pi_{i_h}$.
  \end{itemize}
  Using the definition of $\nu_k$ and the above observation about the entry 
  points, we immediately obtain that, for all $1 \leq \ell \leq n$ 
  and $k > 1$,
  \begin{itemize}
  \item $\nu_k(C_\ell) = \mathit{true}$ \ iff \ 
     $\val(\conf{c_\ell}{k}) = k(n{+}1) - n + 1$;
  \item $\nu_k(C_\ell) = \mathit{false}$ \ iff \  
     $\val(\conf{c_\ell}{k}) = k(n{+}1) + 1$.
  \end{itemize}
  Now, we add a fresh stochastic control state $q_\varphi$ such that
  $q_\varphi \ltran{0} c_\ell$ for every $1 \leq \ell \leq n$. The
  probability of each of these transitions is $1/n$. For every $k \geq 1$
  we have that
  \begin{itemize}
  \item if $\nu_k(C_\ell) = \mathit{true}$, then 
     $\val(\conf{q_\varphi}{k}) = k(n{+}1) - n + 2$;
  \item if $\nu_k(C_\ell) = \mathit{false}$, then at least one
     clause is false, which implies 
     \[
        \val(\conf{q_\varphi}{k}) 
           \quad \geq \quad  
        \frac{n-1}{n} \bigg ( k(n{+}1) - n + 2 \bigg ) 
                + \frac{1}{n} \bigg (k(n{+}1) + 2 \bigg ) 
           \quad = \quad
        k(n{+}1) - n + 3.
    \]
  \end{itemize}
  The construction of $\A$ is completed by adding a non-deterministic
  control state $p$ and a family of stochastic control states
  $d_1,\ldots,d_{n}$, where the transitions are defined as follows
  (here we need that $n \geq 5$):
  \begin{itemize}
  \item $p \ltran{0} c_\varphi$,\quad $p \ltran{0} d_1$,
  \item $d_4 \ltran{-1} d_5$, \quad $d_{n} \ltran{0} p$;,
  \item $d_j \ltran{0} d_{j+1}$ for all $1 \leq j < n$, $j \neq 4$.
  \end{itemize}
  Let $\sigma$ be a pure memoryless strategy in $\M^{\infty}_\A$ such that
  \begin{itemize}
  \item in every configuration of the form $\conf{c_\ell}{k}$, the
    strategy $\sigma$ selects a transition to some optimal entry 
    point for~$k$. If all transitions lead to non-optimal
    entry points, any of them can be selected;
  \item in a configuration of the form $\conf{p}{k}$, the strategy
    $\sigma$ selects either the transition leading to $\conf{q_\varphi}{k}$
    or the transition leading to $\conf{d_1}{k}$, depending on whether
    $\nu_k(\varphi) = \mathit{true}$ or not, respectively.
  \end{itemize}
  Obviously, $\sigma$ is optimal in all configurations of the form
  $\conf{c_\ell}{k}$, and hence it is also optimal in all configurations
  of the form $\conf{q_\varphi}{k}$. By induction on $k$, we show that 
  $\sigma$ is optimal in $\conf{p}{k}$, and $\val(\conf{p}{k})$ equals either 
  $k(n{+}1) -n +3$ or 
  $k(n{+}1)-n+4$, depending on whether $\nu_{k'}(\varphi) = \mathit{true}$
  for some $1\leq k' \leq k$ or not, respectively. 
  \begin{itemize}
  \item $\mathbf{k=1}$. If $\nu_{1}(\varphi) = \mathit{true}$, then
    $\E^\sigma \conf{p}{1} = 4$. Further, it cannot be that
    $\E^{\sigma'} \conf{p}{1} < 4$ for any pure strategy $\sigma'$, because
    \begin{itemize}   
    \item if $\sigma'$ selects the transition from $\conf{p}{1}$ to 
       $\conf{d_1}{1}$, then inevitably $\E^{\sigma'} \conf{p}{1} = 5$;
    \item  if $\sigma'$ selects the transition from $\conf{p}{1}$ to 
       $\conf{q_\varphi}{1}$, then $\E^{\sigma'} \conf{p}{1}$ cannot be 
       less than~$4$ because $\sigma$ plays optimally in $\conf{q_\varphi}{1}$. 
     \end{itemize}
    If $\nu_{1}(\varphi) = \mathit{false}$, then $\E^\sigma \conf{p}{1} = 5$,
    and this outcome cannot be improved by playing the transition from
    $\conf{p}{1}$ to $\conf{q_\varphi}{1}$ because $\sigma$ is optimal
    in $\conf{q_\varphi}{1}$ and $\E^\sigma \conf{q_\varphi}{1} \geq 4$.

    Hence, $\sigma$ is optimal in $\conf{p}{1}$ and $\val(\conf{p}{1})$
    is either $4$ or $5$ depending whether $\nu_{1}(\varphi) = \mathit{true}$
    or not, respectively.
  \item \textbf{Induction step.} Let us consider a configuration
    $\conf{p}{k{+}1}$. If $\nu_{k+1}(\varphi) = \mathit{true}$, then
    $\E^\sigma \conf{p}{k{+}1} = (k{+}1)(n{+}1)-n+3$. Since $\sigma$
    plays optimally in $\conf{q_\varphi}{k{+}1}$, this outcome cannot
    be improved by any pure strategy $\sigma'$ which selects the transition
    from $\conf{p}{k{+}1}$ to  $\conf{q_\varphi}{k{+}1}$. If $\sigma'$
    selects the transition from $\conf{p}{k{+}1}$ to $\conf{d_1}{k{+}1}$,
    then $\conf{p}{k}$ is inevitably reached in exactly $n+1$ transitions.
    By induction hypothesis, this leads to the outcome at least
    $(n{+}1) + k(n{+}1)-n+3 = (k{+}1)(n{+}1)-n+3$. Hence, $\sigma$ is
    optimal and $\val(\conf{p}{k{+}1}) = (k{+}1)(n{+}1)-n+3$.

    If $\nu_{k+1}(\varphi) = \mathit{false}$, then (by applying induction
    hypothesis) $\E^\sigma \conf{p}{k{+}1}$ is equal either
    to $(n{+}1) + k(n{+}1) -n +3$ or to $(n{+}1) + k(n{+}1) -n + 4$, depending
    on whether $\nu_{k'}(\varphi) = \mathit{true}$ for some $1 \leq k' \leq k$
    or not, respectively. In both cases, this yields the desired outcome
    which cannot be improved by using the transition from $\conf{p}{k{+}1}$
    to $\conf{q_\varphi}{k{+}1}$, because then the outcome is inevitably
    at least $(k{+}1)(n{+}1)- n + 4$. 
  \end{itemize}
Now, it suffices to put $K = \prod_{i=1}^m \pi_i$ and $N = K(n{+}1)-n+4$.
Since $\pi_i$ is $\calO(i \log(i))$, the encoding size
of $\A$ is polynomial in $\size\varphi$, and the length of the binary
encoding of $K$ and $N$ is also polynomial in $\size\varphi$.
\qed
\end{proof}

By Lemma~\ref{lem-lower}, the existence of an algorithm which computes 
$\val(\conf{p}{k})$ up to an \emph{absolute} error strictly less than 
$1/2$ in time $\calO(f)$ implies the existence of an algorithm 
for \textsc{SAT} and \textsc{UNSAT} whose time complexity is
$\calO(f \circ p)$, where $p$ is a polynomial. The same can be said 
about an algorithm which computes $\val(\conf{p}{k})$ up to a
\emph{relative} error strictly less than $1/(2 \cdot |Q| \cdot k)$,
where $Q$ is the set of control states of~$\A$. Also note that 
stochastic states in $\A$ have outgoing edges whose probability
is $1$ or $1/n$, but it is trivial to modify the construction so that
all of these probabilities are equal to $1/2$. So, Lemma~\ref{lem-lower}
proves Theorem~\ref{thm-hard} for configurations of the form 
$\conf{q}{i}$. Now we show that we can even take $i=1$.

\begin{figure}[t]
\centering
\begin{tikzpicture}[x=1.3cm,y=1.3cm,font=\scriptsize]
\foreach \i/\j in {4/0,3/1,2/2,1/3,0/4}{%
     \node(q\i) at (\j,0) [ran] {$p_{\i}$};
}
\foreach \i/\j in {4/3, 3/2, 2/1, 1/0}{%
     \draw [->] (q\i) to [tran]  node[above] {$\frac{1}{2}$}  (q\j);
}
\foreach \i/\j in {3/2, 2/1, 1/0}{%
     \draw [->] (q\i) to [tran, bend left=60]  node[above] {$\frac{1}{2}$}  (q4);
}
\draw [->] (q4) to [tran, loop left] node[left] {$\frac{1}{2}$} (q4);
\draw [->] (q0) to [tran, loop right] node[right] {$1$} (q0);
\end{tikzpicture}
\caption{The example gadget $\G_4$.}
\label{fig:long-chain}
\end{figure}
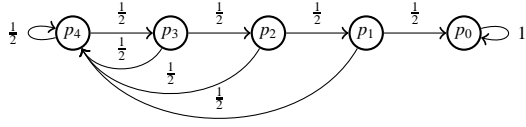

Let us consider the following OC-MDP $\G_k$: the set of control
states is $\{p_0,\ldots,p_k\}$, all of these states are stochastic, 
and there a transition from $p_i$ to $p_{i-1}$ and 
$p_k$ for all $i \geq 1$. All transitions increment the counter 
by $1$ and have probability $\frac{1}{2}$. The
state $p_0$ is a dead-end with a self-loop.
An example for $k=4$ is given in Figure~\ref{fig:long-chain}.

\begin{lemma}
\label{lem:high-prob}
  With probability higher than $\frac{1}{4}$, a run initiated in
  $\conf{p_k}{1}$ visits a configuration $\conf{p_0}{i}$ where
  $i \geq 2^k$.
\end{lemma}
\begin{proof}
  Notice that the probability of terminating in one step is less or equal to
  $2^{-k}$, because in order to reach $p_0$ from $p_k$ the process has to take a
  sequence of $k$ transitions, as otherwise it restarts at $p_k$. Therefore, the
  probability that the process does not reach $p_0$ in $i$ steps is greater or
  equal to $(1-2^{-k})^i$. For $i = 2^k$ we have that this value is
  $(1-2^{-k})^{2^k}$, but it is well-known that the sequence
  $(1-\frac{1}{n})^n$ is increasing in $n$ and converges to $\frac{1}{e}$. As
  for $n=2$ this expression is equal $\frac{1}{4}$, for $k \geq 1$ we get that
  the probability of visiting $p_0$ with the counter value higher 
  than $2^k$ is at least $\frac{1}{4}$.
\qed
\end{proof}

\noindent
We also need the following lemma:
\begin{lemma}
\label{lem:prime-prod}
  $\prod_{i=1}^m \pi_i \leq 2^{m^2}$, where $\pi_m$ is the $m$-th 
  smallest prime number.
\end{lemma}
\begin{proof}
  Of course $\pi_1 = 2$. Bertrand's postulate states that for every $k
  > 1$ there is at least one prime number $p$ such that $k < p <
  2k$. From this we know that there is at least one prime in the
  following disjoint intervals $(2,4)$, $(4,8)$, $(8,16)$, $\ldots$
  which gives us an estimate on the $\pi_i \leq 2^i$. Therefore,
  $\prod_{i=1}^{m} \pi_i \leq \prod_{i=1}^m 2^i = 2^{m(m+1)/2} \leq
  2^{m^2}$ for all $m \geq 1$.  \qed
\end{proof}

\noindent
With the help of Lemma~\ref{lem:high-prob} and Lemma~\ref{lem:prime-prod}, 
we can now prove the following:

\begin{lemma}
\label{lem-increase}  Given a propositional formula $\varphi$ in CNF, 
  one can compute a OC-MDP $\B$ that uses only probabilities 
  $\frac{1}{2}$ on transitions such that being able to approximate 
  $\val(\conf{q}{1})$ up to the absolute error $\frac{1}{8}$ or the
  relative error $2^{-|Q|}$, where $|Q|$ is the number of control
  states of $\B$, suffices to establish whether $\varphi$ is satisfiable
  or not.
\end{lemma}
\begin{proof}
    Let $\varphi$ be an arbitrary CNF formula, we construct a
  polynomially sized OC-MDP $\B$ with probabilities on transitions
  equal $\frac{1}{2}$, such that $\varphi$ is not satisfiable iff the
  optimal termination time from one of the control states and counter
  value $1$ is equal to $(n+2)(2^{m^2+1}-1) - 6$, where $n$ and $m$
  are the number of clauses and variables in $\varphi$,
  respectively. We will build $\B$ by combining the gadget 
  $\G_{m^2}$ (see Fig.~\ref{fig:long-chain}), where $m$ is the number 
  of propositional
  variables in $\varphi$, with the OC-MDP $\A$ that we obtain from
  Lemma~\ref{lem-lower} for $\varphi$. We let the initial state of
  $\B$ be $p_{m^2}(1)$ and the initial control state $p$ of $\A$
  replaces the control state $p_0$ in $\G_{m^2}$. Let $x_k$
  denote the probability that $\A$ will be initiated at $p(k+1)$ in
  $\B$, which is the same as saying that $\A$ executes $k$
  transitions before reaching control state $p$. Of course $\sum_k x_k
  = 1$ and thanks to Lemma~\ref{lem:high-prob} we have $\sum_{k\geq
    2^{m^2}} x_k > \frac{1}{4}.$

  Assume that $\varphi$ is not satisfiable. We know that the expected
  termination time from $p(k)$ in $\A$ is equal to $k(n+1) - n + 4$
  for every $k$, where $n$ is the number of clauses in
  $\varphi$. Therefore $\val(p_{m^2}(1)) = \sum_k x_k \left(k + k(n+1)
    - n + 4\right)$. Let us consider a Markov chain $M$ with positive
  rewards obtained from $\G_{m^2}$ by ignoring the counter completely
  and assigning reward $n+2$ to each transition. Notice that the
  expected total reward before $M$ terminates is equal to $v := \sum_k
  x_k \cdot k(n+2)$, so $\val(p_{m^2}(1)) - v = \sum_k x_k (n - 4) = n
  - 4$. It is quite straightforward to compute $v$ to be
  $(n+2)(2^{m^2+1}-2)$, and so in the end get that $\val(p_{m^2}(1)) =
  (n+2)(2^{m^2+1}-1) - 6$.

  Next, assume that $\varphi$ is satisfiable. Let $k'$ be the smallest
  number such that the assignment to the propositional variables
  corresponding to $k'$ in the proof of Lemma~\ref{lem-lower}
  satisfies $\varphi$. We know that $k' \leq \prod_{i=1}^m \pi_m$
  which is $\leq 2^{m^2}$ thanks to Lemma~\ref{lem:prime-prod}. We
  also know that for all $k < k'$ we have $\val(p(k)) = k(n+1) - n +
  4$ and for all $k \geq k'$ we have $\val(p(k)) = k(n+1) - n +
  3$. Therefore in this case $\val(p_{m^2}(1)) = \sum_{k < k'} x_k
  \left(k(n+1) - n + 4\right) + \sum_{k \geq k'} x_k \left(k(n+1) - n
    + 3\right) = \sum_k x_k (k(n+1) - n + 4) - \sum_{k \geq k'} x_k
  \leq (n+2)(2^{m^2+1}-1) - 6 - \sum_{k \geq 2^{m^2}} x_k \leq
  (n+2)(2^{m^2+1}-1) - 6 - \frac{1}{4}$, where the last step follows
  from Lemma~\ref{lem:high-prob}. Notice that the number of control
  states in $\B$ is $|Q| \geq m^2 + \sum_m \pi_m (n+1)$, so
  $\frac{1}{8} ((n+2)(2^{m^2+1}-1) - 6) \leq 2^{-|Q|}$.  
\qed
\end{proof}

\end{document}